\documentclass[peerreview]{IEEEtran} 
\usepackage{orcidlink}

\usepackage{graphicx}
\usepackage{hyperref}
\usepackage{color}
\hypersetup{
    colorlinks=true,
    linkcolor=black,
    citecolor=black,
    urlcolor=blue
}
\usepackage{wrapfig}

\usepackage{amsmath}    
\usepackage{amsthm}     
\usepackage{amssymb}    
\usepackage{mathtools}
\usepackage{amsfonts}
\usepackage{stmaryrd}
\usepackage{enumerate}
\usepackage{physics}
\usepackage{bbm}        
\usepackage{xparse}
\usepackage{xcolor}
\usepackage{mathrsfs} 
\usepackage{graphicx}
\usepackage{pdfpages}
\usepackage{tcolorbox}
\usepackage{tikz}
\usetikzlibrary{positioning}

\usepackage[nameinlink,capitalize,noabbrev]{cleveref}
\crefname{enumi}{Step}{Steps}

\newtheorem{theorem}{Theorem}[section]

\newtheorem{lemma}[theorem]{Lemma}

\newtheorem{corollary}[theorem]{Corollary}
\newtheorem{remark}[theorem]{Remark}
\newtheorem{definition}[theorem]{Definition}

\numberwithin{equation}{section}

\makeatletter
\newenvironment{proofsketch}{\par\normalfont
\topsep6\p@\@plus6\p@\relax
\trivlist
\item[\hskip\labelsep\itshape
Proof sketch\@addpunct{.}]\ignorespaces
}{
\endtrivlist\@endpefalse
}
\makeatother

\newcommand{\set}[1]{\left\{#1\right\}}
\newcommand{\bits}{\set{0,1}}
\renewcommand{\tr}[1]{\mathrm{Tr}\!\left[ #1 \right]}
\newcommand{\ptr}[2]{\mathrm{Tr}_{#1}\!\left[ #2 \right]}
\newcommand{\pr}[1]{\mathrm{Pr}\!\left[ #1 \right]}
\newcommand{\prs}[2]{\mathrm{Pr}_{#1}\!\left[ #2 \right]}
\newcommand{\proj}[1]{\ketbra{#1}{#1}}

\newcommand{\bO}{\mathrm{O}}
\newcommand{\bbN}{\mathbb{N}}
\newcommand{\1}{\mathbbm{1}}
\newcommand{\ext}{\mathrm{ext}}
\newcommand{\Ext}{\mathrm{Ext}}

\newcommand{\cB}{\mathcal{B}}
\newcommand{\cC}{\mathcal{C}}
\newcommand{\cD}{\mathcal{D}}
\newcommand{\cE}{\mathcal{E}}

\newcommand{\cG}{\mathcal{G}}
\newcommand{\cH}{\mathcal{H}}

\newcommand{\cR}{\mathcal{R}}
\newcommand{\cS}{\mathcal{S}}

\newcommand{\hunp}{\ensuremath{ H_{\mathrm{unp}} }}

\newcommand{\hunps}[1][\default]{H_{\mathrm{unp}\left({#1}\right)}}

\newcommand{\bP}{\Delta_P}

\newcommand{\dist}{\mathrm{d}}

\newcommand{\ind}{\textrm{ind}}
\newcommand{\guess}{\textrm{guess}}
\newcommand{\hill}{\textrm{HILL}}

\newcommand{\CNOT}{\mathrm{CNOT}}

\newcommand{\mch }{\leftrightarrow}
\newcommand{\IP}{\mathrm{IP}}

\usepackage{quantikz} 

\begin{document}

\title{Quantum Computational Unpredictability Entropy \\ and Quantum Leakage Resilience}

\author{%
  Noam Avidan\orcidlink{0009-0009-1893-987X}\IEEEauthorrefmark{1}%
  \and                                           
  Rotem Arnon\orcidlink{0000-0002-5808-6279}\IEEEauthorrefmark{2}%

  \thanks{\IEEEauthorrefmark{1}The Center for Quantum Science and Technology, Faculty of Mathematics and Computer Science, Weizmann Institute of Science, Israel.}%
  \thanks{\IEEEauthorrefmark{2}The Center for Quantum Science and Technology, Faculty of Physics, Weizmann Institute of Science, Israel.}
  }

\maketitle

\begin{abstract}
    Computational entropies provide a framework for quantifying uncertainty and randomness under computational constraints. They play a central role in classical cryptography, underpinning the analysis and construction of primitives such as pseudo-random generators, leakage-resilient cryptography, and randomness extractors. In the quantum setting, however, computational analogues of entropy remain largely unexplored.
    In this work, we initiate the study of quantum computational entropy by defining \emph{quantum computational unpredictability entropy}, a natural generalization of classical unpredictability entropy to the quantum setting. Our definition builds on the operational interpretation of quantum min-entropy as the optimal guessing probability, while restricting the adversary to efficient guessing strategies.
    We prove that this entropy satisfies several fundamental properties, including a leakage chain rule that holds even in the presence of unbounded prior quantum side-information. We also show that unpredictability entropy enables pseudo-randomness extraction against quantum adversaries with bounded computational power. Together, these results lay a foundation for developing cryptographic tools that rely on min-entropy in the quantum computational setting.
\end{abstract}

\section{Introduction}

Classical and quantum notions of entropy, their definitions, properties, and operational meaning are indispensable in cryptography. A prominent example is the conditional min-entropy $H_{\min}(X|E)$, where $X$ is a random variable and $E$ may be a classical random variable correlated with $X$ or even a quantum system. In both cases, the min-entropy quantifies the amount of information that an adversary with access to $E$ has about $X$~\cite{renner2008security,KRS09OperationalMeaningEntropy,renner2004smooth} using an optimal guessing strategy. The conditional min-entropy can therefore be directly related to tasks such as privacy amplification, encryption systems, leakage resilience, and more.

Numerous notions of quantum entropy have been thoroughly studied over the past three decades with great success~\cite{Tomamichel_2016,khatri2024principlesquantumcommunicationtheory}. Quantum cryptography thrives on developments in quantum information theory, and many security proofs build on the mathematical tools that emerge from defining and analyzing entropy measures for quantum systems. Examples of powerful entropy-related results relevant to cryptography include chain rules~\cite{Dupuis_2015_chain_rules_smooth,Fang_2020_chain_rules_relative_entropies,WTHR11Impossibilitygrowingquantumbitcommit}, duality~\cite{Tomamichel_2010_Duality,beigi2013sandwiched_dualiteis,MDSFT13_quantum_renyi_entropies_dualiteis}, the asymptotic equipartition property~\cite{TRR09_fully_quantum_AEP}, entropy accumulation theorems~\cite{DFRR20_EAT,MFSR22_GEAT,AHTE25_GREAT_acumm,arqand2025marginalconstrainedentropyaccumulationtheorem}, decoupling theorems, and more~\cite{dupuis2014oneshotdecoupling,Horodecki_2006_decoupling_negative_info,Tomamichel_2016,khatri2024principlesquantumcommunicationtheory}.
Yet, the vast majority of this work has focused on \emph{information-theoretic} entropy notions, with relatively little attention paid to the \emph{computational} aspects of quantum information theory.

This stands in sharp contrast to the classical setting, where computational entropy has been extensively studied and successfully applied across cryptography. Definitions such as HILL entropy~\cite{haastad1999pseudorandom}, unpredictability entropy, and compression-based entropies~\cite{barak2003computational,HLR07SepPseudoentropyfromCompressibility} have formed the foundation for pseudo-randomness~\cite{haitner2010efficiencyHILLprgOW}, leakage-resilient cryptography~\cite{DP08LeakageResilientStandard,S16_betterHILL_chain_rule}, and randomness extractors~\cite{kalli09extcompassumptions}. One notable prior work~\cite{CC17Computationalminentropy} proposed a quantum variant of HILL entropy, but the framework lacked key structural results, most notably, a general leakage chain rule, and required restrictive assumptions such as bounded quantum storage for cryptographic applications.

At the same time, recent work at the intersection of quantum cryptography, complexity theory, and information theory has introduced a variety of new computationally motivated quantum objects: pseudorandom quantum states~\cite{QPRSJKF18}, EFI pairs~\cite{yan20_quantumcommit_EFI}, pseudorandom unitaries~\cite{bouland2023publickeypseudoentanglementhardnesslearning},  computational pseudoentanglement~\cite{arnonfriedman2023computationalentanglementtheory,aaronson2023quantumpseudoentanglement}, quantum one-way puzzles~\cite {khurana2024commitmentsquantumonewayness}. In all these directions, computational assumptions enter the picture. Through the power of the distinguisher (i.e., the computational model underlying the distance measure), as well as through the complexity of generating or verifying the relevant states or transformations. These studies build on decades of work analyzing the information-theoretic ``non-pseudo'' analogues of these quantum structures using tools from quantum information theory, most notably, entropy.

Despite this, the role of entropy in the emerging theory of quantum computational pseudo-randomness has yet to be fully developed. We argue that a well-founded notion of quantum computational entropy is essential for this endeavor, just as it has been in classical cryptography.

In this work, we take a step in this direction by defining and studying a quantum computational variant of the conditional min-entropy: quantum computational unpredictability entropy. We prove that it satisfies key properties, most notably, a leakage chain rule, and supports the construction of cryptographic primitives such as quantum pseudo-randomness extractors secure against quantum side-information.

\section{Main Contributions and Technical Overview}
The main goal of our work is to advance the (practically non-existent) theory of quantum computational entropy, a necessary step in the foundation for modern quantum cryptography. We introduce a new quantity, which we term \emph{quantum computational unpredictability entropy} $\hunp$, a natural computational variant of min-entropy, and a quantum analog of classical unpredictability entropy~\cite{HLR07SepPseudoentropyfromCompressibility}. Our definition is operationally motivated and aligns with core ideas in quantum information theory.

We establish several properties of $\hunp$, most notably a fully quantum leakage chain rule that holds even in the presence of prior quantum side-information, overcoming a key limitation of the framework introduced in~\cite{CC17Computationalminentropy}. This result provides a powerful tool for reasoning about entropy in cryptographic settings with quantum leakage.

In the following sections, we discuss our contributions and proof techniques in more detail. Some quantum notation is needed in order to explain the results. We present only the necessary background as we go. \cref{sec:prelim} includes a more thorough background.

\subsection{Quantum Computational Unpredictability Entropy}

One of the most widely studied quantum entropies is the conditional min-entropy $H_{\min}(X|E)_{\rho}$~\cite{renner2008security,KRS09OperationalMeaningEntropy}, where $\rho=\rho_{XE}$ is a quantum state.
We are interested in the case when~$X$ is a classical random variable and $E$ is a quantum system. The state can be written as a classical-quantum (cq) state  $\rho_{XE}=\sum_x p(x)\proj{x}\otimes \rho^x_E$, with $\{\ket{x}\}_x$  a family of orthonormal vectors representing the classical values of~$X$. Then, the min-entropy has the following operational meaning~\cite{KRS09OperationalMeaningEntropy},
\begin{equation}\label{eq:min_entropy_intro}
    H_{\min}(X|E)_{\rho} = - \log \mathrm{P}_{\guess}(X|E)\;,
\end{equation}
where $\mathrm{P}_{\mathrm{guess}}(X|E)$ is the optimal probability of guessing the value of $X$ given access to $E$. The optimal way to guess the value of $X$ is by measuring the quantum state and then guessing based on the measurement outcome.\footnote{Formally we write $\mathrm{P}_{\guess}(X|E)=\mathbb{E}_x \tr{E_x\rho^x_E}$ where $E_x$ are positive operator-valued measures (POVMs) $\{E_x\}_x$, calculated on the side-information $\rho_E$, and the expectation and measurement outcomes are defined via the cq-state $\rho_{XE}$.}
The quantum measurements achieving the optimal guessing probability are also relevant for questions in quantum hypothesis testing, their study dates back to~\cite{holevo1973statistical,yuen1975optimum}.

In this work, we introduce a \emph{computational variant} of the quantum min-entropy given in~\cref{eq:min_entropy_intro}. That is, the guessing strategies are now limited in their computational power.
A classical counterpart of such an entropy was introduced in~\cite{HLR07SepPseudoentropyfromCompressibility} and termed the ``unpredictability entropy''. Our quantity of interest is, therefore, both a quantum extension of the classical unpredictability entropy and the computational variant of the information-theoretic quantum min-entropy.

In its simplest form, one can define the quantum unpredictability entropy as follows:
Given a cq-state $\rho_{XE}$ we say that
\begin{definition}
    $\hunps(X|E)_{\rho}\ge k$
    if for any quantum guessing circuit $\cC$ of size $s$, $\pr{\cC(\rho_{E}^x) = x} \le 2^{-k} \;$.
\end{definition}
By limiting the size of the quantum circuit $\cC$ we limit the allowed guessing strategies and, hence, this acts as an extension of the operational definition of the min-entropy to a setup in which computational complexity matters.\footnote{The definition above is restricted to the case of cq-states, i.e., when $X$ is a classical register. In a follow-up work, soon to appear on the arXiv, we show how to extend the definition also to fully quantum states, where both systems may be quantum~\cite{avidan2025fully}.}

A key advantage of this definition is that it allows us to directly capture the uncertainty associated with computational hardness, something min-entropy cannot express. For instance, if $F$ is a post-quantum cryptographic hard to invert permutation, and $X$ is a uniformly random input, then the min-entropy $H_{\min}(X|F(X))$ is zero, since $X$ is fully determined by $F(X)$. However, the unpredictability entropy $\hunps(X|F(X))$ can still be high, reflecting the fact that $F$ is computationally hard to invert. This highlights the operational nature of our definition: it quantifies the success probability of any \emph{efficient} guessing strategy, making it directly applicable to cryptographic settings.

We now wish to ``smooth'' the entropy, as typically done in quantum information theory. Meaning, instead of evaluating the entropy on a given state $\rho_{XE}$, we allow some flexibility and optimize the value of $\hunps(X|E)$ over all states $\tilde{\rho}_{XE}$ that are $\varepsilon$-close to $\rho_{XE}$.
The distance measure with which one chooses to define closeness matters. In the classical world, the statistical distance and its computational analog are mostly used. The quantum extension of the statistical distance is the so-called trace distance, and a related computational version is also easy to define (see~\cref{sec:prelim}).
Those distance measures were used to define the classical computational entropies~\cite{haastad1999pseudorandom,barak2003computational,HLR07SepPseudoentropyfromCompressibility} as well as the quantum HILL-entropy of~\cite{CC17Computationalminentropy}.

When dealing with quantum entropies, however, a more adequate distance measure used to define smooth entropies is the purified distance~\cite{Tomamichel_2010_Duality}, which we here denote by $\bP$.\footnote{We postpone giving the formal definition of the purified distance to~\cref{def:pur_dist} below.}
Using the purified distance, we suggest the following definition:
\begin{definition}\label{def:unpredictability-entropy_intro}
    Given a cq-state $\rho_{XE}$, we say that
    $\hunps^{\varepsilon}(X|E)_{\rho}\ge k$,
    if there exists a (sub-normalized) cq-state~$\tilde{\rho}_{XE}$ such that~$\bP(\rho_{XE},\tilde{\rho}_{XE})\leq\varepsilon$, and for any quantum guessing circuit $\cC$ of size at most $s$, $\pr{\cC(\tilde{\rho}_{E}^x) = x}\le2^{-k}\;$.
\end{definition}

Note that in the above definition, the size of the guessing circuit $\cC$ is bounded by $s$. The size of a circuit or distinguisher that defines the distance measure $\bP$ and asserts that $\bP(\rho_{XE},\tilde{\rho}_{XE})\leq \varepsilon$, however, is \emph{un}bounded.
This distinction marks a key departure from the classical setting: in classical computational entropy notions such as unpredictability entropy~\cite{HLR07SepPseudoentropyfromCompressibility}, the smoothing is performed with respect to a \emph{computational} distance measure (e.g., indistinguishability by bounded-size circuits), and not with respect to a statistical or information-theoretic one. As a result, classical unpredictability entropy automatically \emph{upper bounds} the HILL entropy, and thus assigns high entropy to the output of pseudorandom generators.

\cref{def:unpredictability-entropy_intro} is by itself fundamental and relevant for applications in cryptography.
In particular, we show that the new computational entropy fulfills a quantum leakage chain-rule (in contrast to the computational notions of min-entropy~\cite{CC17Computationalminentropy}) and can be used to quantify how much pseudo-randomness can be extracted using a randomness extractor.

Indeed, one of our main contributions is a quantum leakage chain-rule for the quantum computational unpredictability entropy.
\begin{theorem}\label{thm:chain_intro}
    For any quantum state $\rho_{XBC}$, classical on $X$, and any $\varepsilon\ge 0$, $s\in\bbN$, $\ell = \log\dim(C)$, we have:
    \begin{equation}\label{eq:chain_rule_intro}
        \hunps^{\varepsilon}(X|BC)_{\rho} \ge \hunps[s + \bO(\ell)]^{\varepsilon}(X|B)_{\rho} - 2\ell \;.
    \end{equation}
\end{theorem}

The factor of 2 accompanying $\ell$ in~\cref{eq:chain_rule_intro} is fundamentally quantum (it can be seen as a consequence of quantum superdense coding~\cite{chen2017leakagesuperdense}) and tight in general.
The above chain-rule is the quantum equivalent of the leakage chain-rules for classical computational entropies~\cite{DP08LeakageResilientStandard,reingold2008dense,fuller2012computational} and the computational counterpart of the quantum information-theoretic leakage chain-rule~\cite{WTHR11Impossibilitygrowingquantumbitcommit}. Indeed, the proof of the quantum information-theoretic chain-rule can be retrieved when $s\rightarrow\infty$ (unlimited computational power) and the classical chain-rule can be proven in the same way but with classical registers (diagonal in the computational basis) and the appropriate dimension factors (i.e., instead of $2\ell$ one can easily get $\ell$ in the classical case). We say that our chain rule is fully quantum in the sense that both $B$ and $C$ are quantum registers. This is in contrast to the chain rule proven in~\cite{CC17Computationalminentropy}, which required $B$ to be classical and was therefore limited to the quantum bounded-storage model.
The generality of our chain rule allows us to move beyond this restriction: we can handle adversaries that hold arbitrary quantum side-information about the secret $X$, and that repeatedly leak quantum information via general bounded-dimension quantum channels. Our results, therefore, apply to a much broader class of leakage scenarios than those captured by prior models.

\subsection{Extracting Pseudo-Randomness in the Presence of a Quantum Adversary}
Our next contribution is both of fundamental nature and of relevance for applications.
It is well known that quantum-proof extractors can be used to extract randomness from a source $X$ of high min-entropy $H_{\min}^{\varepsilon}(X|E)_{\rho}$, with the adversary holding the quantum system $E$.
Does high $\hunps^{\varepsilon}(X|E)_{\rho}$ imply that pseudo-randomness can be extracted from $X$? We show that at least for some extractors the answer is positive.
Formally, a quantum-proof extractor is defined as follows:
\begin{definition}\label{def:ext_intro}
    A function $\Ext: \bits^{n} \times \bits^{d} \to \bits^{m}$ is a quantum proof $(\varepsilon_{\ext},k)$ strong extractor if for all ccq-states $\rho_{XYE}$, such that $X\in\bits^{n}$ with $H_{\min}(X|E) \ge k$ and $Y\in\bits^{d}$ a uniform random seed:
    \begin{equation*}
        \dist(\rho_{\Ext(X,Y)YE},\rho_{U^m}\otimes\rho_{YE}) \le \varepsilon_{\ext}\;,
    \end{equation*}
    with $\dist(\cdot,\cdot)$ standing for the trace distance and $\rho_{U^m}$ is maximally mixed state over $m$ bits.\footnote{The trace distance is the quantum extension of the statistical distance, and $\rho_{U^m}$ is the quantum notation for a random variable distributed uniformly over~$\bits^m$.}
\end{definition}

Intuitively, since the extractor works with $H_{\min}(X|E)_{\rho}$ (and thus even better with $H_{\min}^{\varepsilon}(X|E)_{\rho}$), it should also work with $\hunps^{\varepsilon}(X|E)_{\rho}$, an adversary which is computationally limited can only do worse than an unbounded one.
When considering the HILL-entropy, the answer is indeed trivially yes; This simply follows from the definition of the entropy (see~\cref{def:hill_cq-entropy}).
As we saw, however, the HILL-entropy does not have a fully quantum leakage chain-rule, for example, and thus we still wish to consider our new entropy.
Regrettably, the situation is more complex when considering the unpredictability entropy.
Even classically, only certain randomness extractors are known to work with unpredictability entropy, namely, reconstructive extractors with efficient reconstruction~\cite{HLR07SepPseudoentropyfromCompressibility}. Extending these results to the quantum setting presents additional challenges. We discuss the unique challenges of extractors in the quantum computational setting in~\cref{subec:ChallengesExtQuantumUnp}.

To answer our question in the quantum world, we first go back to fundamental results in the study of quantum-proof extractors.
One of the most renowned results~\cite{KTB08boundedstorageext} states that any single-bit output (i.e., fixing $m=1$ in~\cref{def:ext_intro}) randomness extractor works in the presence of a quantum adversary, with a small difference in the parameters compared to a classical adversary.
A main ingredient in the proof is a technique called the ``pretty good measurement''~\cite{hausladen1994pretty}, originally developed for quantum hypothesis testing. The complexity of the quantum algorithm implementing the measurement depends on the complexity of the state of the adversary and is, in general, high~\cite{gilyen2022quantum} and, hence, it is not clear that the proof of the soundness of extractors in the information-theoretic case~\cite{KTB08boundedstorageext} can be extended to our setup, in which computational complexity matters.

We thus revert to studying the case of an explicit simple randomness extractor, the inner-product (IP) extractor.
We prove the following theorem:
\begin{theorem}\label{thm:IP_intro}
    Let $\rho_{XE}$ be a cq-state where $X$ is distributed over $\bits^{n}$ and $Y$ be uniformly distributed over~$\bits^{n}$. Let $k_{\ext} \in \mathbb{N}, \varepsilon_{\ext}>0$ and $k_{\ext} \ge 1-2\log(\varepsilon_{\ext})$. We denote by $\IP(X,Y)$ the binary inner-product of the values taken by $X$ and $Y$.
    If
    \begin{equation*}
        \hunps[2s+2n+5]^{\varepsilon}(X|E) \ge k_{\ext} \;,
    \end{equation*}
    then
    \begin{equation*}
        \dist_{s}(\rho_{\IP(X,Y)YE},\rho_{U_{1}}\otimes\rho_{YE}) \le \varepsilon_{\ext} + 2\varepsilon \;.
    \end{equation*}
\end{theorem}

To prove the above theorem, we employ proof techniques from~\cite{KK10twosourceextractorssecurequantum}, originally used to show that the IP is a good two-source extractor against bounded-storage quantum adversaries.
The relevance of~\cite{KK10twosourceextractorssecurequantum} to our work lies in the fact that~\cite{KK10twosourceextractorssecurequantum} uses the operational meaning of the min-entropy, i.e., a quantifier of the optimal guessing strategy of the source; Specifically, it follows from~\cite{KK10twosourceextractorssecurequantum} that if the IP extractor is not secure, then one can derive a good guessing strategy for the \emph{initial} string.
Since we are interested in a seeded extractor, in contrast to a two-source extractor (with two entangled quantum adversaries), we can simplify the proof and show that in the case of a perfect seed, it also works without a bound on the storage of the adversary.
Crucially, the reduction that shows that if the IP extractor is not secure, the initial source can be guessed with higher probability than assumed, is constructive; The guessing strategy is explicit and \emph{efficient}. Therefore, not only a lower-bound assumed on the min-entropy is broken, but the same holds true for our quantum computational unpredictability entropy $\hunp$. The IP extractor outputs one bit; to extend it to get many bits, we follow the proofs developed in~\cite{dodis2004improved,KK10twosourceextractorssecurequantum,DPVR12trevisanextwithquantumsideinfo}, and analyze the computational resources required for each step. Notably, the extension to many bits requires the use of the leakage chain rule to bound how much a short advice string can reduce the entropy, as we show in \cref{sec:ps_rand}.

Our work opens many new questions in this context. Are there better ways to extract pseudo-randomness from sources $\hunps^\varepsilon(X|E)$ in the presence of a quantum adversary? For example, are there extractors for $\hunps^\varepsilon(X|E)$ with a special reconstruction property~\cite{barak2003computational}? If so, they can be combined with Trevisan's extractor~\cite{trevisan2001extractors,DPVR12trevisanextwithquantumsideinfo} to create pseudo-randomness with an initial logarithmic seed. Extending the result of~\cite{DPVR12trevisanextwithquantumsideinfo} to show that Trevisan's extractor also works with $\hunps^\varepsilon(X|E)$ is somewhat tedious but not too challenging technically. Finding a single-bit extractor that (1) has the reconstruction property required for Trevisan's extractor and (2) works with the $\hunps^\varepsilon(X|E)$, however, seems to be harder.
Classically, explicit list-decodable codes exist and are known to work~\cite{barak2003computational,trevisan2001extractors}. Quantumly, previous work builds on the mentioned work~\cite{KTB08boundedstorageext} regarding single-bit extractors, but, as explained, it is not clear how to extend the result to the computational setup.

\subsection{Alternating Extraction in the Presence of Quantum Leakage}
As an additional application of our newly defined entropy and the results above, we analyze alternating extraction protocols in the presence of quantum leakage. Alternating extraction is a well-known cryptographic technique for deriving fresh random bits from independent weak sources using public seeds. In our setting, we show that such protocols remain secure even when each round leaks bounded quantum information to the adversary. This requires a leakage chain rule for the cqq case, to bound the reduction of unpredictability entropy of the source under repeated quantum leakages, or when the source is already correlated with prior quantum side-information.

Our result generalizes the classical alternating extraction analysis of~\cite{DP08LeakageResilientStandard} to the quantum setting. Specifically, we consider alternating application of seeded extractors to two independent sources, where after each round of extraction, a bounded amount of quantum information may leak to the adversary.
We show that the unpredictability entropy of the sources degrades by a controlled amount at each step and that pseudo-random bits can still be securely extracted using certain quantum-proof extractors. The analysis crucially relies on the quantum leakage chain rule established in \cref{subsec:chain_rule_for_unp} and on our results about pseudo-randomness extraction from unpredictability entropy in \cref{sec:IP}.

To formalize this model, we extend the classical ``Only Computation Leaks'' (OCL) model~\cite{OCLmodel_04} to the quantum setting, allowing for quantum side-information and leakage.

\begin{definition}\label{def:bounded-quantum-leakage-channel_intro}
    Let $\rho_{XE}$ be a cq-state. A channel $\phi : \cS_{\circ}(XE) \to \cS_{\circ}(XLE')$ is called a $\lambda$-\emph{bounded quantum leakage channel} for $\rho_{XE}$ from if it can be written as a composition $\phi = \Lambda \circ \psi$, where:
    \begin{enumerate}
        \item $\psi : \cS_{\circ}(E) \to \cS_{\circ}(E')$ is a pre-processing CPTP map acting only on the adversary's system $E$ (it does not access~$X$),
        \item $\Lambda:\cS_{\circ}(XE') \to \cS_{\circ}(LXE')$ is a CPTP map that first appends an ancilla register $L$ in the state $\proj{0}_L$, of dimension at most $2^{\lambda}$, and then modifies only the $L$ register, leaving the $XE'$ marginal invariant:
              \begin{equation*}
                  \ptr{L}{\Lambda(\rho_{XE'})} = \rho_{XE'} \;.
              \end{equation*}
    \end{enumerate}
\end{definition}

In the above definition, one should think of $X$ as the classical data on which the computation at a given round is being performed. The state $\rho_{LE'}$ is the adversary's state after the new information was leaked, with the dimension of~$L$ (standing for leakage) being bounded. $\rho_{LE'}$ does not have to be of the form~$\rho_{L}\otimes\rho_{E'}$, which would correspond to the case where the leakage is an independent system given to the adversary in every round.

This leakage model allows for a large class of leakage attacks, and it is relevant in both the computational and information-theoretic settings.
Between leakage rounds, the adversary can perform arbitrary quantum operations on the side information. The leakage itself is restricted to be ``read only'' on the state $\rho_{XE'}$. The requirement that the leakage channel may not change the state $X$ is natural, as attacks that actively modify the secret are much stronger than leakage attacks. The requirement that the side-information can not change during the leakage, accept via the new ancilla~$L$, may seem restrictive, as general quantum operations may change their inputs, this requirement appears to be necessary. If we drop it without replacement, it becomes impossible to guarantee meaningful entropy after leakage. An explicit attack in such a setting is described in \cref{subsec:alternating-extraction}.

The decomposition is also useful when turning to the computational setting. It is natural to restrict the computational power of the adversary between leakage rounds. The decomposition lets us impose this restriction while allowing for the leakage channels themselves to be of unbounded complexity, restricted only by the number of qubits that can be leaked in each round.

Using this model, we prove that alternating application of extractors remains secure under quantum leakage. Our proof tracks the evolution of unpredictability entropy across rounds using the leakage chain rule (\cref{thm:chain_intro}) and shows that pseudo-randomness can still be extracted in each round via \cref{thm:IP_intro}, and the extension to multi-bit output as described in \cref{sec:IP}.

\subsection{Relation to Previous Works}

As mentioned in the introduction, in contrast to the study of computational entropies in classical cryptography~\cite{haastad1999pseudorandom,barak2003computational,HLR07SepPseudoentropyfromCompressibility,S16_betterHILL_chain_rule,haitner2010efficiencyHILLprgOW}, quantum computational entropies were only considered from a cryptographic perspective in a single seminal work~\cite{CC17Computationalminentropy}. Recent work defined a new variant of computational entropy motivated by complexity-constrained thermodynamics and quantum hypothesis testing~\cite {Yunger_Halpern_2022,Munson_2025}.

The main contribution of~\cite{CC17Computationalminentropy} is the definition and analysis of a quantum variant of the HILL-entropy. The suggested entropy, however, did not fulfill many properties that one would expect to have, such as a leakage chain-rule. This limited the usage of the quantum HILL-entropy in cryptographic applications. Their model for quantum leakage was unsatisfactory, as it required adding the bounded-storage assumption, and did not account for leakage channels where new leakage may be entangled with prior quantum side-information.

The complexity entropy defined in~\cite{Munson_2025} has similar properties to previously defined computational entropies and, under a conjectured chain rule, an operational meaning related to data compression; it is not clear if it is directly comparable to any of the quantum computational entropies previously defined in~\cite{CC17Computationalminentropy} or in this work.

We took a different angle by defining a new computational entropy, the computational unpredictability entropy. Our definition extends both the classical unpredictability entropy and the quantum smooth min-entropy.
The leakage chain-rule that we prove (\cref{thm:chain_intro}) is an extension of both the leakage chain-rules for classical computational entropies~\cite{HLR07SepPseudoentropyfromCompressibility,KPWW14CounterexampleHILLchainrule} to quantum leakage and the computational extension of the quantum information-theoretic leakage chain-rule~\cite{WTHR11Impossibilitygrowingquantumbitcommit}.

Working with our entropy, while proving the desired properties (such as a leakage chain-rule and the ability to extract pseudo-randomness from unpredictability) allowed us to overcome fundamental difficulties that arose in~\cite{CC17Computationalminentropy}.
Our model of quantum OCL (\cref{def:bounded-quantum-leakage-channel_intro}) is more general than the model used in~\cite{CC17Computationalminentropy}, and we believe it is better motivated in terms of the understanding of quantum process.

Our work opens many new research directions of different flavors, from pure quantum information theory, through questions about quantum codes and randomness extraction, to quantum cryptography. We list some of the questions in \cref{sec:open_questions}.

\section{Preliminaries}\label{sec:prelim}
We assume some familiarity with standard notation in quantum information theory. For completeness and consistency, we include here the definitions we use in this work. For a comprehensive introduction to quantum information theory, we refer the reader to one of several books on the subject, such as~\cite{Tomamichel_2016,nielsen2010quantumbook,wilde2013quantumbook}.

\subsection{Basic Quantum Notation}
We work in finite-dimensional Hilbert spaces. $\ket{\phi}$ denote a vector in a Hilbert space and $\bra{\phi}$ the complex conjugate of it.

\begin{definition} 
    A quantum state is a positive semi-definite matrix with trace $\leq1$.
    A pure quantum state is a state with a matrix of rank $1$. Pure states can be written in the form $\ket{\phi}\bra{\phi}$ for some vector $\ket{\phi}$. States that are not pure are called mixed states.
    Any mixed state can be written as a convex combination of pure states
    \begin{equation*}
        \rho = \sum_{i} p_{i} \proj{\rho_{i}} \;,
    \end{equation*}
    where $\set{p_{i}}$ is a probability distribution and $\proj{\rho_{i}}$ are pure states.
    We say a state is normalized if $\tr{\rho} = 1$ and subnormalized if $\tr{\rho} \le 1$. We denote the sets of all normalized states on a Hilbert space $\cH_{A}$ by $\cS_{\circ}(A)$ and the set of all subnormalized states by $\cS_{\bullet}(A)$.
\end{definition}

\begin{definition}[Classical-Quantum (CQ) States]
    A classical-quantum (cq) state is a state of the form
    \begin{equation*}
        \rho_{XE}=\sum_{x} p_{x} \proj{x}\otimes \rho^{x}_{E} \;,
    \end{equation*}
    with $\set{\ket{x}}_x$  being the standard basis vectors in $\cH_{X}$, representing the classical values of~$X$.
\end{definition}

\begin{definition}\label{def:max_mix}
    We say a state is maximally mixed if it is of the form
    \begin{equation*}
        \omega = \frac{1}{\dim(A)} \sum_{i} \proj{i} \;,
    \end{equation*}
    where $\ket{i}$ is the standard basis of the Hilbert space $\cH_{A}$.
\end{definition}

A state is said to be bipartite or multipartite if the Hilbert space it is acting on is a tensor product space of two or more Hilbert spaces. We denote the Hilbert space of a system $A$ by $\cH_{A}$, and the Hilbert space of a system $B$ by $\cH_{B}$, and the Hilbert space of the composite system $AB$ by $\cH_{AB} = \cH_{A} \otimes \cH_{B}$. Similarly, for the states themselves, the subscript indicates the Hilbert space $\rho_{AB}\in \cS_{\bullet}(AB)$.

\begin{definition}
    We say a bipartite state is maximally entangled if it is of the form $\proj{\Phi}$ where
    \begin{equation*}
        \ket{\Phi} = \frac{1}{\sqrt{\dim(A)}} \sum_{i} \ket{i} \otimes \ket{i} \;.
    \end{equation*}
\end{definition}

\begin{definition}
    A quantum channel is a completely positive trace preserving (CPTP) map. Channels map quantum states to quantum state.
    We use the following notation for channels acting on states:
    \begin{equation*}
        \rho_{\phi(A)B} := (\phi_{A} \otimes \1_{B}) (\rho_{AB})\;,
    \end{equation*}
    to denote the state of the system after applying the channel $\phi$ on the marginal~$\rho_{A}$.
\end{definition}

\begin{definition}
    Positive Operator-Valued Measure (POVM) are generalized measurements that can be performed on quantum states.
    POVMs can be modeled as a set of positive semi-definite Hermitian matrices $\set{E_{i}}$ such that $\sum_{i}E_{i} = \1$. The probability of outcome $i$ on a state $\rho$ is $\tr{E_{i}(\rho)}$.
\end{definition}

Any quantum channel followed by any measurement can be modeled as a POVM.
A POVM can be modeled by a channel operating on the state and some auxiliary system, followed by a measurement.

\subsection{Distance Measures}


We now present a few distance measures that we use throughout the paper. These measures quantify how distinguishable two quantum states are, and play a crucial role in defining our computational entropy. We begin with the trace distance, a quantum analog of statistical distance.
\begin{definition}[Trace Distance]
    The trace distance between two quantum states $\rho$ and $\sigma$ is defined as
    \begin{equation*}
        \dist(\rho,\sigma) = \frac{1}{2} \norm{\rho - \sigma}_{1} = \frac{1}{2} \tr{\sqrt{(\rho - \sigma)^{\dagger}(\rho - \sigma)}} \;.
    \end{equation*}
\end{definition}

Purified distance, derived from fidelity, provides a metric that is particularly useful in the context of smoothing quantum entropies.
One key property of the purified distance is the following definition of fidelity.
\begin{definition}[Fidelity]
    The fidelity between states $\rho_{A},\sigma_{A}$ can be stated in terms of maximal overlap between purifications:
    \begin{equation*}
        F(\rho_{A},\sigma_{A}) = \max_{\ket{\phi_{AB}},\ket{\psi_{AB}}} (\abs{\braket{\phi_{AB}}{{\psi_{AB}}}})^2\;,
    \end{equation*}
    where the maximum is taken over all pure states such that
    $\rho_{A}=\ptr{B}{\proj{\phi_{AB}}}$, and  $\sigma_{A}=\ptr{B}{\proj{\psi_{AB}}}$.
\end{definition}

We define a generalized fidelity to work with sub-normalized states.

\begin{definition}[Generalized Fidelity]\label{Uhlmanstheorem}
    The Generalized Fidelity between subnormalized states $\rho_{A},\sigma_{A}$:
    \begin{equation*}
        F_{*}(\rho_{A},\sigma_{A}) = \left(\max_{\ket{\phi_{AB}},\ket{\psi_{AB}}} \abs{\braket{\phi_{AB}}{{\psi_{AB}}}}+ \sqrt{(1-\tr{\rho_{A}})(1-\tr{\sigma_{A}})}\right)^{2}\;,
    \end{equation*}
    where the maximum is taken over all pure states such that
    $\rho_{A}=\ptr{B}{\proj{\phi_{AB}}}$, and  $\sigma_{A}=\ptr{B}{\proj{\psi_{AB}}}$.
\end{definition}

Note that if at least one of the states is normalized, the generalized Fidelity and the Fidelity are the same.

The equivalence between this definition of fidelity and other definitions of fidelity is known as Uhlmann's theorem~\cite{uhlmann1976transition}.

\begin{definition}[Purified Distance~\cite{Tomamichel_2010_Duality}]\label{def:pur_dist}
    The purified distance between sub-normalized states $\rho$ and $\sigma$ is:
    \begin{equation*}
        \bP(\rho,\sigma) = \min_{\ket{\phi_{AB}},\ket{\psi_{AB}}} \sqrt{1 - F_{*}(\ket{\phi_{AB}},\ket{\psi_{AB}})}  \;,
    \end{equation*}
    where the minimum is taken over all pure states such that
    $\rho_{A}=\ptr{B}{\proj{\phi_{AB}}}$, and  $\sigma_{A}=\ptr{B}{\proj{\psi_{AB}}}$.
\end{definition}

\begin{lemma}[Uniqueness of Fidelity~\cite{tan2024prospectsdeviceindependentquantumkey}]\label{lem:uniq_fidlity_for_pure_uhlmann}
    Let $G: \cS_{\bullet}(A) \times \cS_{\bullet}(A) \to \mathbb{R}$ be a real-valued function on states over a Hilbert space $A$. Suppose $G$ satisfies both of the following properties:
    \begin{enumerate}
        \item \textbf{Data-Processing Inequality:} For any (CPTP) map $ \mathcal{M} $ and $\rho, \sigma \in \cS_{\bullet}(A)$,
              \begin{equation*}
                  G\left( \mathcal{M}(\rho), \mathcal{M}(\sigma) \right) \ge G(\rho, \sigma) \;.
              \end{equation*}

        \item \textbf{Pure-Uhlmann Property:} For all $\rho, \sigma \in \cS_{\bullet}(A) $, letting $\ket{\psi}, \ket{\phi}$ range over all purifications of $\rho$ and~$\sigma$, we have
              \begin{equation*}
                  G(\rho, \sigma) = \max_{\ket{\psi},\ket{\phi}}  G\left(\ket{\psi}, \ket{\phi} \right)\;.
              \end{equation*}
    \end{enumerate}
    Then there exists a monotonic increasing function $g : \left[0,1\right] \to \mathbb{R}$ such that
    \begin{equation*}
        G(\rho, \sigma) = g\left( F(\rho, \sigma) \right) \;.
    \end{equation*}
\end{lemma}
Replacing the maximum with a minimum in the pure-Uhlmann property, and reversing the direction of the data-processing inequality, yields an equivalent lemma in which the function $g$ is monotonically decreasing.

Therefore, fidelity is essentially the unique function that satisfies both properties above for normalized states. For a detailed proof and discussion, see~\cite[Appendix H.1]{tan2024prospectsdeviceindependentquantumkey}.
We believe that this is a key part in the difficulty in proving chain-rules for quantum entropies based on indistinguishability, like the quantum (relaxed) HILL entropy as defined in~\cite{CC17Computationalminentropy}. In the limit case of unbounded computational power ($s\to \infty$), computational indistinguishability is the trace distance. Due to the lack of extension property to the trace distance, a ``smooth min-entropy'' with smoothing based on the trace distance dose not admit a fully quantum leakage chain-rule.

As a consequence, the purified distance is the most natural distance measure for quantum states in any context where both the data-processing inequality and the pure-Uhlmann property are required. As we will show in~\cref{sec:main_def_res}, several desirable properties of smooth quantum computational entropies depend on these two properties, for example, the existence of a well-defined dual entropy. We explore this dual quantity, including its operational interpretation, properties, and a generalization to fully quantum states, in our follow-up work~\cite{avidan2025fully}. Other results, such as~\cref{lem:chain-rule-quantum-unpredictability-entropy}, rely on a weaker property of the purified distance, we refer to as Uhlmann's extension property, see~\cref{lem:extension_uhlmann_property_for_purefied_distance} for more details.

\subsection{Entropies}

Entropies are fundamental quantities in information theory and cryptography, used to quantify the uncertainty or randomness of a system. Here, we define min-entropy, a measure of the worst-case randomness of a quantum state, and quantum-proof seeded extractors, which are essential tools for extracting randomness from weak sources.

\begin{definition}[Min-Entropy~\cite{renner2006securityquantumkeydistribution}]
    Let $\rho_{AB}$ be a bipartite quantum state. The conditional min-entropy of $A$ given $B$ is defined as:
    \begin{equation*}
        H_{\min}(A|B)_{\rho} =  \sup_{\lambda \in \mathbb{R} ,\sigma_{B} \in \cS_{\bullet}(B) } \left\{ \lambda : \rho _{AB} \le 2^{-\lambda} (\mathbb{I}_A \otimes \sigma_B) \right\}\;.
    \end{equation*}
\end{definition}

\begin{definition}
    For a state $\rho\in \cS_{\bullet}(A)$ and $\sqrt{\tr{\rho}}>\varepsilon\ge0$ we define a $\varepsilon-$purified ball around $\rho$ as:
    \begin{equation*}
        \cB_{\varepsilon}(\rho) = \set{\tilde{\rho}\in \cS_{\bullet}(A):  \text{ s.t. } \bP(\rho,\tilde{\rho})\le\varepsilon}  \;.
    \end{equation*}
\end{definition}

\begin{definition}[Smooth Min-Entropy~\cite{RennerKonig2005_smoothentext}]
    Let $\rho_{AB}$ be a bipartite quantum state and $\varepsilon\ge 0$. The conditional $\varepsilon$ smooth min-entropy of $A$ given $B$ is defined as:
    \begin{equation*}
        H_{\min}^{\varepsilon}(A|B)_{\rho} = \sup_{\tilde{\rho}\in\cB_{\varepsilon}(\rho)}H_{\min}(A|B)_{\tilde{\rho}}\;.
    \end{equation*}
\end{definition}

\begin{definition}[Quantum Proof Seeded Extractor]\label{def:quantum_proof_seeded_ext}
    A function
    \begin{equation*}
        \Ext: \bits^{n} \times \bits^{d} \to \bits^{m} \;,
    \end{equation*}
    is a quantum-proof $(\varepsilon_{\ext},k)$ strong extractor if for all  ccq-states $\rho_{XYE}$, such that
    $H_{\min}(X|E)\ge~k$, and let~$Y\in\bits^{d}$ be a uniform and independent random seed:
    \begin{equation*}
        \dist(\rho_{\Ext(X,Y)YE},\rho_{U^m}\otimes\rho_{YE}) \le \varepsilon_{\ext}\;,
    \end{equation*}
    where $\dist(\cdot,\cdot)$ is the trace distance and $\rho_{U^m}$ is maximally mixed state over $m$ bits.
\end{definition}


\subsection{Quantum Computational Model}

\begin{definition}[Quantum Circuits]
    We fix a finite set of universal elementary quantum gates. A quantum circuit is a sequence of quantum gates that act on a set of qubits, and measurements in the computational basis. The size of a circuit is the number of gates in the circuit. 
\end{definition}

\begin{remark}
    We work with a fixed universal gate set~$\mathcal{G}$. For simplicity, we assume every $G\in\mathcal{G}$ acts non-trivially on at most two qubits, and that $\mathcal{G}$ is closed under taking inverses, i.e., $G^\dagger\in\mathcal{G}$ for all $G\in\mathcal{G}$.
    Concrete choices such as ${\mathrm{H},\mathrm{T},\mathrm{CNOT}}$ or ${\mathrm{Clifford}\cup\mathrm{T}}$ satisfy these conditions.
\end{remark}

\begin{definition}[Distinguisher]
    A channel $\cC:\cS_{\circ}(A) \to \bits$ is called a distinguisher. We denote the set of all distinguishers with a circuit size of at most $s$ by $\cD_{s}$.
\end{definition}

\begin{definition}[Computational Distance]
    Let $\rho,\sigma$ be states in the same Hilbert space. Let $s\in\bbN$, the $s$-computational distance between $\rho$ and $\sigma$ is:
    \begin{equation*}
        \dist_s(\rho,\sigma) = \sup_{\substack{\cC \in \cD_{s}}} \abs{\pr{\cC(\rho)=1} - \pr{\cC(\sigma)=1}} \;.
    \end{equation*}
\end{definition}
Intuitively, the $s$-computational distance $\dist_s(\rho,\sigma)$ measures the maximum distinguishing advantage that any quantum circuit of size at most $s$ can achieve between $\rho$ and $\sigma$. We say that $\rho$ and $\sigma$ are $(s,\varepsilon)$-computationally indistinguishable if $\dist_s(\rho,\sigma) \le \varepsilon$.

\section{Quantum Computational Unpredictability Entropy}\label{sec:main_def_res}

\subsection{Definition and Basic Properties}\label{sec:unp_def}

We suggest a definition of quantum unpredictability entropy that combines the operational meaning of a computationally bounded guessing circuit with the information-theoretic purified distance.

\begin{definition}[Quantum Computational Unpredictability Entropy]\label{def:purified-unpredictability-entropy}
    For any cq-state $\rho_{XE}$, and $\varepsilon\ge 0$, $s\in\bbN$.
    We say that
    \begin{equation*}
        \hunps^{\varepsilon}(X|E)_{\rho}\ge k \;,
    \end{equation*}
    if there is a cq-state $\tilde{\rho}_{XE} \in \cB_{\varepsilon}\left(\rho_{EX}\right)$ such that for any guessing circuit $\cC$ of size $s$
    \begin{equation*}
        \pr{\cC(\tilde{\rho}_{E}^x) = x} \le 2^{-k} \;.
    \end{equation*}
\end{definition}

\begin{remark}
    The definition above is well-defined only when $X$ is classical, as the guessing probability is inherently a classical concept.
    The use of the purified distance ensures that, whenever $\rho_{XE}$ is a cq-state, there exist suitable cq-states $\tilde{\rho}_{XE}$ within the $\varepsilon$-ball around it, i.e., $\tilde{\rho}_{XE} \in B_\varepsilon(\rho_{XE})$, that can be used in the smoothing optimization~\cite{Tomamichel_2010_Duality}.
    Note that such smoothed states $\tilde{\rho}$ may be sub-normalized.

    A natural question is whether one can define a meaningful notion of computational unpredictability entropy when $X$ is also quantum, i.e. for fully quantum states. We address this in a follow-up work~\cite{avidan2025fully}, where we develop such a generalization and explore its operational meaning.
\end{remark}
We now state a few basic properties of the quantum conditional unpredictability entropy.
\begin{lemma}[Monotonicity]
    For any $\varepsilon' \ge \varepsilon \ge 0$
    \begin{equation*}
        \hunps^{\varepsilon'}(X|E)_{\rho} \ge \hunps^{\varepsilon}(X|E)_{\rho}  \;.
    \end{equation*}
    For any $s' \ge s$
    \begin{equation*}
        \hunps^{\varepsilon}(X|E)_{\rho} \ge \hunps[s']^{\varepsilon}(X|E)_{\rho}  \;.
    \end{equation*}
\end{lemma}
For cq-states, smooth min-entropy can be defined as the maximal guessing probability using \emph{any} quantum circuit~\cite[Theorem 1]{KRS09OperationalMeaningEntropy}, implying the following lemma:
\begin{lemma}\label{lem:lim_unp_is_smoothmin}
    Let $\rho_{XE}$ be a cq-state, $s\in\mathbb{N},\varepsilon\ge 0$
    \begin{align*}
        \hunps^{\varepsilon}(X|E)_{\rho}                   & \ge H_{\min}^{\varepsilon}(X|E)_{\rho}  \;, \\
        \lim_{s\to\infty} \hunps^{\varepsilon}(X|E)_{\rho} & = H_{\min}^{\varepsilon}(X|E)_{\rho}  \;.
    \end{align*}
\end{lemma}

\begin{lemma}[Data-Processing Inequality]\label{lem:data_processing_inequality_unpredictability_entropy}
    Let $\rho_{XE}$ be a cq-state, let $s\in\mathbb{N}$, $\varepsilon\ge0$ and let $\Phi_{E\to E'}$ be a quantum channel that can be implemented using a circuit of size $t$,
    \begin{equation*}
        \hunps^{\varepsilon}(X|E')_{\Phi(\rho)} \ge \hunps[s+t]^{\varepsilon}(X|E)_{\rho}\;.
    \end{equation*}
\end{lemma}
Monotonicity and its relations to the smooth min-entropy follow directly from the definitions. For completeness, a detailed proof of the data-processing inequality is provided in~\cref{apendix_sec_more_proofs}.

Smooth min-entropy has a dual quantity called smooth max-entropy. Smooth max-entropy is related to several operational tasks in quantum information, such as entanglement distillation, decoupling, state merging, and data compression~\cite{KRS09OperationalMeaningEntropy}.
We define a dual quantity to our quantum unpredictability entropy:

\subsection{Chain Rule with Unbounded Quantum Side-Information}\label{subsec:chain_rule_for_unp}
Leakage chain rules are a useful tool for conditional entropy, they allow us to bound how much new information reduces the entropy. In the information-theoretic setting, the leakage chain rule is known for smooth min-entropy from~\cite{WTHR11Impossibilitygrowingquantumbitcommit}, with no degradation in the smoothing parameter.
In the classical computational setting, the leakage chain rule for unpredictability entropy is well-known and was shown in~\cite[Lemma 11]{KPWW14CounterexampleHILLchainrule}.

Our proof builds upon the core idea of the classical proof: bounding the probability that a random guess for the leakage $C$ would be correct. However, adapting this intuition to the quantum setting requires some technical effort. Instead of relying on inequalities of probabilities, our proof leverages inequalities of positive operators, a crucial adaptation for handling quantum side-information. This blend of classical intuition with quantum techniques allows us to establish a robust leakage chain rule in the quantum computational setting.

\begin{theorem}\label{lem:chain-rule-quantum-unpredictability-entropy}
    For any cqq-state $\rho_{XBC}$, classical on $X$, and any $\varepsilon \ge 0$, $s\in\bbN$, $\ell = \log\dim(C)$, we have:
    \begin{equation*}
        \hunps^{\varepsilon}(X|BC)_{\rho} \ge \hunps[s+\bO(\ell)]^{\varepsilon}(X|B)_{\rho} - 2\ell \;.
    \end{equation*}
\end{theorem}
For the proof, we need the following lemma from~\cite{WTHR11Impossibilitygrowingquantumbitcommit}.
\begin{lemma}\label{lem:inequality-of-operators-for-unp-leakage-chain-rule}
    For any state $\rho_{A}$ and any extension $\rho_{AB}$, we have:
    \begin{equation*}
        \rho_{AB} \le \dim(B)^{2} (\rho_{A} \otimes \omega_{B}) \;,
    \end{equation*}
    where $\omega_{B}$ is the maximally mixed state on $B$, recall~\cref{def:max_mix}.
\end{lemma}
This is a special case of the pinching inequality~\cite {Masahito_Hayashi_2002_pinching}. For completeness, we provide detailed proof in our notations in~\cref{apendix_sec_more_proofs}.

Additionally, we need a property of the purified distance, closely related to Uhlmann's theorem, that was mentioned in~\cref{Uhlmanstheorem}, and~\cref{lem:uniq_fidlity_for_pure_uhlmann}.
\begin{lemma}[Extension Property for the Purified Distance]\label{lem:extension_uhlmann_property_for_purefied_distance}
    For any $\rho_{AB},\sigma_{A}$ there is an extension of $\sigma_{A}$, $\ptr{B}{\sigma_{AB}}=\sigma_{A}$ such that the purified distance is the same,
    \begin{equation*}
        \bP(\rho_{AB},\sigma_{AB}) = \bP(\rho_{A},\sigma_{A}) \;.
    \end{equation*}
\end{lemma}

\begin{remark}
    The same is not true for trace distance.\footnote{
        To see that it is not true for trace distance, we can look at $\rho_{AB},\sigma_{A}$ such that $\rho_{A}$ is maximally mixed and $\sigma_{A}$ is pure, both on one qubit, $\dist(\rho_{A},\sigma_{A})= \frac{1}{2}$ but for any purifications $\sigma_{AB},\rho_{AB}$  $\dist(\rho_{AB},\sigma_{AB})\ge \frac{1}{\sqrt{2}}$, since $\sigma_{AB}$ is a product state between $A$ and $B$ and $\rho_{AB}$ is maximally entangled on the same partition.}
    For classical distributions, an analogues property naturally holds for statistical distance. If $X$ and $Y$ are close distributions, for any joint distribution $XZ$ there is a joint distribution $YW$ that is as close, in statistical distance.
\end{remark}

With the tools established above, the quantum leakage chain rule follows from a direct and clean argument. The proof mirrors the classical case~\cite{KPWW14CounterexampleHILLchainrule}, using the positivity of POVMs and the extension property of the purified distance to lift the classical reduction to the quantum setting. As expected from the connection to superdense coding, a factor of~$2$ naturally emerges.

\begin{proof}[Of~\cref{lem:chain-rule-quantum-unpredictability-entropy}]
    By contraposition, we will show that:
    \begin{equation*}
        \hunps^{\varepsilon}(X|BC)_{\rho} < k - 2\ell \implies \hunps[s+\bO(\ell)]^{\varepsilon}(X|B)_{\rho} < k \;.
    \end{equation*}

    Assume $\hunp^{s}(X|BC)_{\rho} < k - 2\ell$. Let $\tilde{\rho}_{XBC}$ such that $\bP(\tilde{\rho}_{XBC},\rho_{XBC})\le \varepsilon$, there is a guessing circuit $\cC$ of size $s$, with corresponding POVM $\set{E^{x}_{BC}}_{x}$ such that:
    \begin{equation*}
        \sum_{x} \tilde{p}(x)\tr{E^{x}_{BC}\tilde{\rho}^{x}_{BC}} > 2^{-k+2\ell} \;.
    \end{equation*}
    From the extension property of the purified distance~\cref{lem:extension_uhlmann_property_for_purefied_distance} we know that any pair of states  $\tilde{\rho}_{XB}$ and $\rho_{ABC}$ such that $\bP(\tilde{\rho}_{XB},\rho_{XB})\le\varepsilon$ there is an extension such that
    \begin{equation*}
        \bP\left(\tilde{\rho}_{XBC},\rho_{XBC}\right)\le \varepsilon \;.
    \end{equation*}
    We know from~\cref{lem:inequality-of-operators-for-unp-leakage-chain-rule}, that for any $x$:
    \begin{equation*}
        \tilde{\rho}^{x}_{BC} \le \dim(C)^{2} (\tilde{\rho}^{x}_{B} \otimes \omega_{C}) \;.
    \end{equation*}
    By definition $\dim{C} = 2^{\ell}$, so we can rewrite it as:
    \begin{equation*}
        (\tilde{\rho}^{x}_{B} \otimes \omega_{C})  \ge 2^{-2\ell} \tilde{\rho}^{x}_{BC} \;.
    \end{equation*}
    A circuit that gets $\tilde{\rho}^{x}_{B}$ with probability $p(x)$, can first generate $\tilde{\rho}^{x}_{B} \otimes \omega_{C}$ using $\bO(\ell)$ gates, and then apply the same guessing circuit with the POVM $\set{E^{x}_{BC}}_{x}$ on~$\tilde{\rho}^{x}_{B} \otimes \omega_{C}$ to guess $x$.

    Since POVMs are positive operators, we can apply them to both sides of the inequality, using the contrapositive assumption, and the extension property of the purified distance, we get:
    \begin{align*}
        \sum_{x} \tilde{p}(x)\tr{E^{x}_{BC} \tilde{\rho}^{x}_{B} \otimes \omega_{C}} \ge 2^{-2\ell} \sum_{x} \tilde{p}(x) \tr{E^{x}_{BC} \tilde{\rho}^{x}_{BC}} > 2^{-2\ell} \cdot 2^{-k+2\ell} \;.
    \end{align*}
    Therefore we have that $\hunps[s+\bO(\ell)]^{\varepsilon}(X|B)_{\rho} < k$, which concludes the proof of the chain rule:
    \begin{equation*}
        \hunps^{\varepsilon}(X|BC)_{\rho} \ge \hunps[s+\bO(\ell)]^{\varepsilon}(X|B)_{\rho} - 2\ell \;. \qedhere
    \end{equation*}
\end{proof}

\begin{corollary}
    In the limit $s\to\infty$, we recover the proof for the chain rule for smooth min-entropy for cq-states.
\end{corollary}

\begin{proof}
    Recall that in the limit $s\to \infty$, unpredictability entropy and smooth min-entropy are equivalent~\cref{lem:lim_unp_is_smoothmin}.
    For any finite $\ell$, the number of additional gates needed, $L_{g} = \bO(\ell)$ is a fixed finite number, the limit stays the same $\lim_{s\to \infty} s = \lim_{s\to \infty} s + L_{g}$. In the limit both sides of the inequality are smooth min-entropy, with the same smoothing parameter $\varepsilon$. We thus recover the known leakage chain rule for smooth min-entropy for cq-states~\cite{WTHR11Impossibilitygrowingquantumbitcommit}. %
\end{proof}

Our proof can also be modified to recover the classical chain rule for unpredictability entropy~\cite{KPWW14CounterexampleHILLchainrule}.
Recall that for classical distributions $XBC$, the leakage chain rule is:
\begin{equation*}
    \hunps^{\varepsilon}(X|BC)_{\rho} \ge \hunps[s+\bO(\ell)]^{\varepsilon}(X|B)_{\rho} - \ell \;.
\end{equation*}
If we restrict all systems to be classical, our proof recovers a leakage chain rule for classical unpredictability with an unnecessary factor of $2$.\footnote{For quantum side-information the factor of $2$ is fundamental and tight. It can be seen as a consequence of quantum superdense coding~\cite{chen2017leakagesuperdense}.}
This can be corrected by replacing the operator inequality from~\cref{lem:inequality-of-operators-for-unp-leakage-chain-rule} with the following inequality for classical probability distributions:
\begin{equation*}
    \pr{AB = ab} \le |B| \pr{A=a} \pr{U_{B} = b} \;,
\end{equation*}
where $U_B$ is uniform over the support of $B$.

In addition, the smoothing in the classical setting is typically done using statistical or computational distance, which, unlike trace distance for quantum states, also satisfies a natural extension property for classical distributions. Thus, both technical components of our quantum proof carry over to the classical case in a simplified form, allowing the chain rule to be recovered exactly.

\section{Pseudo-Randomness from Unpredictability Entropy}\label{sec:ps_rand}

In this section, we present a method for extracting pseudo-randomness from distributions with high unpredictability entropy. First, we demonstrate how to extract a single pseudo-random bit using the inner-product extractor. Next, we show how to construct extractors that output multiple pseudo-random bits from a single-bit extractor, using weak designs. Finally, we explain why we use the inner-product extractor specifically rather than more general methods.

\subsection{Pseudo-Randomness Extractors}\label{sec:IP}

\begin{definition}[Seeded Extractors from Unpredictability Entropy]\label{def:ext_pseudorandomness_from_unp}
    A function $\Ext:\bits^{n}\times \bits^{d}\to \bits^{m}\times\bits^{d}$ is a~$(k,\varepsilon,\varepsilon',s,s')$-seeded extractor for quantum unpredictability entropy if for any cq-state $\rho_{XU_{d}Q}$ such that the marginal state~$\rho_{U_{d}}$ is maximally mixed and independent of all other registers, and
    \begin{equation*}
        \hunps^{\varepsilon'}(X|Q)_{\rho} \ge k \;,
    \end{equation*}
    the output is $(\varepsilon+2\varepsilon',s)$ indistinguishable from uniform randomness given~$Q$
    \begin{equation*}
        \dist_{s'}(\rho_{\Ext(XU_{d})Q},\rho_{U_{m}U_{d}Q} ) \le \varepsilon+2\varepsilon' \;.
    \end{equation*}
\end{definition}

The information that the adversary holds may be unbounded in both dimension and computational complexity. We only require sufficiently high unpredictability entropy, meaning it is hard to guess $X$ using the side-information and a quantum computer with bounded computational power.

Following the analysis of~\cite{KK10twosourceextractorssecurequantum} and~\cite{CDNT98quantumentanglementcommunicationcomplexityinnerproduct} we will show that inner-product is a good extractor for quantum unpredictability entropy.

\begin{theorem}[Inner-Product Extractor from Unpredictability]\label{lem:IPgoodQExt}
    Let $\rho_{XE}$ be a cq-state where $\rho_{X}$ is a distribution over $\bits^{n}$. Let $\rho_{Y}$ be maximally mixed over $n$ qubits. Let $k_{\ext} \in \mathbb{N}, \varepsilon_{\ext}>0$ such that $k_{\ext} \ge 1-2\log(\varepsilon_{\ext})$
    \begin{equation*}
        \hunps[2s+2n+5]^{\varepsilon}(X|E) \ge k_{\ext} \implies \dist_{s}(\rho_{\IP(X,Y)YE},\rho_{U_{1}}\otimes\rho_{YE}) \le \varepsilon_{\ext} + 2\varepsilon \;.
    \end{equation*}
\end{theorem}

The proof builds on~\cite{KK10twosourceextractorssecurequantum}.
We give here a sketch of the proof.
The formal proof that includes the modifications compared to~\cite{KK10twosourceextractorssecurequantum} is presented in the appendix~\cref{lem:inner_prod_aprox_toX_apendix}.

\begin{proofsketch}
    Let $\rho_{XE}$ be a cq-state. Let $\cE$ be an adversary who can distinguish the inner-product $\IP(x,y)$, for any $y$ from a uniformly random bit with probability at least $\varepsilon$ using a circuit of size $s$, it can predict the inner-product with probability at least $\frac{1}{2}+\varepsilon$, using~\cref{lemma:one_bit_distinguish_is_predict}.

    Assuming that an adversary can predict the inner-product with probability at least $\frac{1}{2}+\varepsilon$ using a circuit of size $s$, then there is a circuit of size as most $2s+2n+6$ that can predict all of $x$ with probability at least~$4\varepsilon^{2}$.

    Therefore, the inner-product is a good  $\left(1-2\log(\varepsilon),\varepsilon\right)$ seeded extractor against quantum side-information and quantum unpredictability entropy. %
\end{proofsketch}

Formally, the proof uses~\cref{lemma:one_bit_distinguish_is_predict} and the following lemma:
\begin{lemma}\label{lem:fromIPtoX}
    Let $\rho_{XE}$ be a cq-state.
    Let $\cC$ be a circuit of size $s$ that can guess $\IP(x,y)$ using $\rho_{E}^{x}$ with probability $\frac{1}{2} + \varepsilon$, where the probability is over the distribution of $x$ and a uniformly random $y$.
    There is a circuit $\cG$ of size at most $2(s+n+3)$ that can guess $x$ using $\rho_{E}^{x}$ with probability at least $4\varepsilon^{2}$.
\end{lemma}

The next part is constructing $m$ bit extractors out of $1$ bit extractors. Showing that this construction is secure for quantum unpredictability entropy results in a computational version of~\cite[Theorem 4.6]{DPVR12trevisanextwithquantumsideinfo}, combining weak $(t,r)$-design with 1-bit extractors.

\begin{definition}[Weak $(t,r)$-Design~\cite{RAZDesigns_extractors_02}]
    A family of sets $S_{1},\dots,S_{m} \subset [d]$ is a weak $(t,r)$-design if for all $i$:  $\abs{S_{i}} = t$ and $\sum_{j=1}^{i-1} 2^{\abs{S_{i}\cap S_{j}}} \le rm$.
\end{definition}
The key idea of a weak $(t,r)$-design is that each seed-block $S_i$ overlaps any earlier block in only a few bits, so that across $m$ one-bit extractions, the total overlap, and hence the entropy reduction by the advice, grows only as $rm$.

\begin{theorem}\label{lem:m-bit-ext-from-any-1-bit-and-design}
    Let $C':\bits^{n}\times\bits^{t} \to \bits$ be a $(k,\varepsilon)$-one-bit extractor secure against $s$-unpredictability entropy.
    Let $S_{1},\dots,S_{m} \subset [d]$ be a weak $(t,r)$-design. Define the following function:
    \begin{align}\label{eq:ext_from_one_bit_ext_and_design}
        \Ext_{C}: \bits^{n} \times \bits^{d} & \to \bits^{m}                                                                 \\
        (x,y)                                & \mapsto \left(C(x,y_{S_{1}}) \;,\dots \;, C(x,y_{S_{m}}) \right)\;, \nonumber
    \end{align}
    where $y_{S}$ is the bits of $y$ in locations $S$.
    $\Ext_{C}$ is a $\left(k + rm - 8\log(\varepsilon), 2 m \varepsilon\right)$ extractor of pseudorandom bits for quantum unpredictability entropy in the following sense: If
    \begin{equation*}
        \hunps[s']^{\varepsilon'}(X|E)_{\rho} \ge k + rm - 8\log(\varepsilon) \;,
    \end{equation*}
    then
    \begin{equation*}
        \dist_{s}(\rho_{\Ext_{C}(X,Y)YE}, \rho_{U_{m}} \otimes \rho_{Y} \otimes \rho_{E}) \le 2 m \varepsilon + 2\varepsilon' \;,
    \end{equation*}
    where $s' = \bO(ns+rm)$.
\end{theorem}

The proof closely follows the approach in~\cite{DPVR12trevisanextwithquantumsideinfo}, with full details provided in~\cref{sec:appendix_trevisan}. The argument relies primarily on two tools: the equivalence, for computationally bounded adversaries, between distinguishing a bit from uniform and predicting it; and the triangle inequality for computational distance. 

\begin{lemma}\label{lemma:one_bit_distinguish_is_predict}
    Let $\rho_{XE}$ be a cq-state where $X$ is classical over one bit. For any~$s\in\mathbb{N}$:
    \begin{equation*}
        \dist_{s} (\rho_{XE},\rho_{U_{1}}\otimes\rho_{E} ) = \dist_{s} (p_{0}\rho_{E}^{0},p_{1}\rho_{E}^{1}) \;.
    \end{equation*}
    Equivalently, there is a circuit of size at most $s$ that on input $\rho_{E}$ correctly guesses $\rho_{X}$ with probability at least
    \begin{equation*}
        \frac{1}{2}+\dist_{s}(p_{0}\rho_{E}^{0},p_{1}\rho_{E}^{1}) \;.
    \end{equation*}
\end{lemma}

\begin{proof}
    We can write the states as:
    \begin{align*}
        \rho_{XE}                   & = p_{0} \proj{0} \rho_{E}^{0} + p_{1} \proj{1} \rho_{E}^{1} \;,                                                                    \\
        \rho_{U_{1}}\otimes\rho_{E} & = \frac{1}{2} \proj{0} ( p_{0}\rho_{E}^{0} + p_{1}\rho_{E}^{1}) + \frac{1}{2} \proj{1} (p_{0}\rho_{E}^{0} + p_{1}\rho_{E}^{1}) \;.
    \end{align*}
    We note that a distinguishing circuit that measures the first bit in the computational basis now needs to distinguish between the post-measurement states. There are two cases:
    Assume the circuit measured $0$, the computational distance for the post-measurement state is:
    \begin{equation*}
        \dist_{s} \left( p_{0}\rho_{E}^{0} , \frac{1}{2} (p_{0}\rho_{E}^{0} + p_{1} \rho_{E}^{1} )\right)\;.
    \end{equation*}
    Since $\rho_{E}^{0}$ cannot be distinguished from itself, we can rewrite it as:
    \begin{align*}
        \dist_{s} \left( p_{0}\rho_{E}^{0} , \frac{1}{2} (p_{0}\rho_{E}^{0} + p_{1} \rho_{E}^{1} )\right)
         & =   \abs{\pr{\cC(p_{0}\rho_{E}^{0})=1} - \pr{\cC\left(\frac{1}{2}(p_{0}\rho_{E}^{0} + p_{1}\rho_{E}^{1})\right)=1} } \\
         & = \frac{1}{2} \abs{\pr{\cC(p_{0}\rho_{E}^{0})=1} - \pr{\cC(p_{1}\rho_{E}^{1})=1}}                                    \\
         & =  \dist_{s} \left(p_{0}\rho_{E}^{0}, p_{1}\rho_{E}^{1} \right)
    \end{align*}
    When the distinguisher measured $1$:
    \begin{equation*}
        \dist_{s} \left( p_{1} \rho_{E}^{1} , \frac{1}{2} (p_{0}\rho_{E}^{0} + p_{1} \rho_{E}^1 )\right) =  \frac{1}{2} \dist_{s} \left(p_{1}\rho_{E}^{1}, p_{0}\rho_{E}^{0} \right) \;.
    \end{equation*}
    A distinguishing circuit can always measure the first classical bit for free and perfectly distinguish between~$0$ and $1$. From the deferred measurement principle, it can start with that measurement and then condition the rest of the operations on the result with no loss in circuit size. Therefore
    \begin{equation*}
        \dist_{s} (\rho_{XE},\rho_{U_{1}}\otimes\rho_{E}) = \dist_{s} (p_{0}\rho_{E}^{0},p_{1}\rho_{E}^{1}) \;. \qedhere
    \end{equation*}
\end{proof}

\begin{lemma}[Triangle Inequality for Computational Distance]\label{lemma:triangle-inequality-computational-td}
    For any $s\in\bbN$ and states $\rho,\sigma,\tau$:
    \begin{equation*}
        \dist_{s}(\rho,\sigma) \le \dist_{s}(\rho,\tau) + \dist_{s}(\tau,\sigma) \;.
    \end{equation*}
\end{lemma}
For completeness, a detailed proof of the triangle inequality can be found in~\cref{apendix_sec_more_proofs}.

\begin{lemma}[Lemma 15~\cite{RAZDesigns_extractors_02}]\label{lem:designs-t-r}
    For every $t,m,r\in \mathbb{N}$ there exists a weak $(t,r)$-design $S_{1},\dots,S_{m}\subset [d]$ such that $d = t \left\lceil \frac{t}{\ln r} \right\rceil $.
    Moreover, such a design can be found in time $\mathrm{poly}(m,d)$ and space $\mathrm{poly}(m)$.
\end{lemma}

Combining the inner-product one-bit extractor from~\cref{lem:IPgoodQExt} with the general reduction from m bits to 1~\cref{lem:m-bit-ext-from-any-1-bit-and-design} and the weak designs construction from~\cref{lem:designs-t-r}, with $r=n^\gamma$ for some $0<\gamma<1/2$ we can construct a seeded extractor that outputs multiple pseudorandom bits from sources with quantum unpredictability entropy.

\begin{lemma}[Extracting More Pseudo-Randomness from Unpredictability Entropy]\label{lem:m-bit-ext-from-unp-IP}
    Let $\varepsilon_{\ext} > 0$,  $n\in \mathbb{N}$ and $0<\gamma<\alpha \leq 1$. 
    There exist $m=n^{\alpha-\gamma}-o(1)$, $d=\bO\left( n^{2}/\log(n) \right)$,
    $k_{\ext} = n^\gamma m + 8\log(m) + 8 \log(\varepsilon_{\ext})+\bO(1)$ and $S_{1},\dots,S_{m}\subset [d]$ such that \begin{align*}
        \Ext:
        \bits^{n} \times \bits^{d} & \to \bits^{m}                                                           \\
        (x,y)                      & \mapsto \left(\IP(x,y_{S_{1}})\;, \dots \;, \IP(x,y_{S_{m}})\right) \;,
    \end{align*}
    is an $(\varepsilon_{\ext},k_{\ext})$-seeded extractor secure against quantum side information and unpredictability entropy, satisfying 
    \begin{equation*}
        \hunps[s']^{\varepsilon}(X|E)\ge k_{\ext} \implies \dist_{s} (\rho_{\Ext(X,U_{d}) U_{d}E},\rho_{U_{m}U_{d}E} ) < \varepsilon_{\ext} + 2\varepsilon \;,
    \end{equation*}
    with $s' = \bO(ns+m)$.
\end{lemma}

\begin{proof}
    Following the modular proof structure of Trevisan's extractors as shown in~\cite{DPVR12trevisanextwithquantumsideinfo} and extended to the computational setting in~\cref{sec:appendix_trevisan}. We use the inner-product extractor~\cref{lem:IPgoodQExt}, as the one-bit extractor in~\cref{lem:m-bit-ext-from-any-1-bit-and-design} with $(n,n^\gamma)$ weak design results in a good seeded extractor against quantum unpredictability entropy. The existence of weak $(n,n^\gamma)$ designs can be seen from~\cref{lem:designs-t-r}. Combining the results above, we get the relation between the parameters $\varepsilon_{\ext}$ and $k_{\ext},d$ as stated in the lemma.
    Note that the big O notation includes both the $2n+2s+5$ from the inner-product unpredictability degradation, as well as the $m$ bits to $1$ bit reduction degradation. %
\end{proof}

\subsection{Challenges of Extracting from Quantum Unpredictability}\label{subec:ChallengesExtQuantumUnp}
In the previous section, we demonstrated how to construct multi-bit extractors from the inner-product extractor. We now turn to the question of why we specifically use the inner-product extractor, and what challenges arise when attempting to use more general extraction methods in the context of quantum unpredictability.


\subsubsection{\hill~vs. Unpredictability}
We can see via a simple hybrid argument that any extractor that extracts almost uniform bits from sources with high min-entropy also extracts pseudorandom bits from sources with high \hill-entropy. This is true whether the side-information is classical or quantum.

The same hybrid argument does not work for unpredictability entropy. To prove that an extractor is secure against unpredictability entropy, we require a stronger argument:
if the output of the extractor can be efficiently distinguished from uniform randomness with sufficiently high probability, then the input can be guessed \emph{efficiently}, with sufficiently high probability.

\subsubsection{Quantum Side-Information and Reconstruction}
In the classical setting it is known that such extractors exist, such extractors are sometimes called reconstructive extractors~\cite{barak2003computational} and they can be constructed from any list decodable code. A \emph{reconstructive extractor} has the special property that any efficient distinguisher for its output can be turned into an efficient \emph{reconstructor}. Given a short “advice” string, one can recover the entire source.  This reconstruction guarantee is exactly what lets unpredictability-based entropy bound the extractor's security.

\begin{definition}\label{def:reconstruction-property-classical}
    A $(L,\varepsilon,s_{dec})$-reconstruction for a function $$E:\bits^{n}\times \bits^{d}\to \bits^{m}\times\bits^{d}$$ is a pair of Turing machines $C,D$ such that:
    $C:\bits^{n}\to\bits^L$ a randomized Turing machine, and $D^{(\cdot)}:\bits^{L}\to\bits^{n}$, a randomized Turing machine that run in total time $s_{dec}$, such that for any $x\in\bits^{n}$ and any distinguisher $T$, if
    \begin{equation*}
        \abs{\pr{T(E(x,U_{d}))=1} - \pr{T(U_{m}\times U_{d})=1}} > \varepsilon \;,
    \end{equation*}
    then, if $D$ has oracle access to $T$, with probability over the distribution of $x\in X$ and the randomness of all the randomized Turing machines:
    \begin{equation*}
        \pr{D^{T}(C(x))=x}>\frac{1}{2}\;.
    \end{equation*}
\end{definition}
In~\cite{trevisan2001extractors} it was implicitly shown that any such function with $(\ell,\varepsilon)$-reconstruction is a $(\ell-\log(\varepsilon),\varepsilon)$-extractor.
Lemma 5 in~\cite{HLR07SepPseudoentropyfromCompressibility} shows that such extractors extract pseudo-randomness from sources with unpredictability entropy.

Note that in this classical definition, the reconstruction may use the distinguisher on the same side-information many times.
For quantum side-information this is not necessarily possible.
It is possible that one can use some quantum side-information to distinguish the output from uniform, but only by measuring it and destroying the quantum information.

In~\cite{DPVR12trevisanextwithquantumsideinfo}, reconstructive extractors are shown to be secure against quantum side-information in the information-theoretic setting. Part of the proof is similar to our proof, reducing from $m$ bits to $1$ bit extractors.

However, the $1$ bit part of the proof relies on the fact that any one-bit extractor is also secure against quantum side-information~\cite{KTB08boundedstorageext}. The problem is that the proof for general one-bit extractors from~\cite{KTB08boundedstorageext} relies on the ``pretty good measurement''~\cite{hausladen1994pretty}. Namely, to guess $\rho^{x}$ using the side-information $\rho_{E}$ they apply the following pretty good measurement with POVM elements:
\begin{equation*}
    F_{x} = P_{X}(x) \rho_{E}^{-1/2} \rho^{x} \rho_{E}^{-1/2} \;.
\end{equation*}
The complexity of this measurement depends on the side-information $\rho_{E}$, which in our setting is unbounded. In this work, we are looking for computational bounds that are independent of the size of the side information the adversary may hold.


\subsubsection{Smoothing with Purified Distance}
As we discussed before, \hill-entropy is defined by smoothing using computational distance. This makes some hybrid arguments to translate information-theoretical results to computational results. However, the computational distance lacks some key properties in the quantum setting, such as the extension property of quantum purified distance~\cref{lem:extension_uhlmann_property_for_purefied_distance}.
In contrast, our quantum smooth unpredictability entropy is defined using the purified distance $\bP$, which is \emph{not} computational. This enables us to prove desirable properties, such as the leakage chain rule, but comes at the cost of no longer satisfying this relationship to HILL entropy. In particular, our unpredictability entropy does \emph{not} assign high entropy to the output of a pseudorandom generator (PRG) evaluated on a random seed. Indeed, a bounded adversary can sample a random seed and evaluate the PRG efficiently.
The probability of this process resulting in a correct guess is related to the entropy of the seed and not the size of the image.
As a result, the unpredictability entropy of $\mathrm{PRG}(\text{seed})$ remains at most the seed length $\abs{\text{seed}}$, despite the fact that the output may be pseudorandom in the traditional cryptographic sense.
This fundamental limitation underscores the gap between unpredictability-based and indistinguishability-based notions of computational entropy, especially in the quantum setting.

As a consequence of this limitation, our definition of unpredictability entropy does not support the construction of leakage-resilient stream ciphers. Such constructions typically rely on repeatedly amplifying entropy between leakage events using a pseudorandom generator (PRG). Since unpredictability entropy does not increase under PRG expansion, this step fails in our framework. Nonetheless, our work does provide a rigorous foundation for cryptographic protocols in the presence of quantum side-information and repeated quantum leakage. In particular, we analyze alternating extraction protocols that remain secure despite cumulative quantum leakage, using the leakage chain rule and pseudo-randomness extraction results developed in this work.

One may further extend the definition of quantum unpredictability entropy by introducing a new distance measure that would correspond to a \emph{computational purified distance}. There are several non-trivial ways to define such a notion, but as far as we are aware, no fully satisfactory or ``natural'' definition currently exists. The origin of this difficulty is tied to the fact that Uhlmann's theorem~\cite{uhlmann1976transition} (see also~\cref{lem:uniq_fidlity_for_pure_uhlmann}) makes no reference to the computational complexity of switching between purifications. This gives rise to what is now called the complexity of the Uhlmann transformation~\cite{bostanci2023unitary}. One possible direction might be to define a variant $\bP^O$ that allows oracle access to a circuit implementing such a transformation. We elaborate on this open question in \cref{sec:open_questions}.

We believe that the difficulty of defining a computationally meaningful distance measure that retains the desirable properties of purified distance, the inherent hardness of the Uhlmann transformation, and the lack of a fully quantum leakage chain rule for quantum HILL entropy are all intimately connected. The purified distance has the crucial advantage of lifting smoothing across purifications, a property not shared by trace or computational distances. However, its non-computational nature makes it incompatible with the indistinguishability framework used in HILL entropy. Our unpredictability-based approach avoids this obstacle by directly bounding guessing success probabilities. This difference helps explain why unpredictability entropy admits a fully quantum leakage chain rule, while the same result for HILL entropy remains out of reach.


\section{Alternating Extraction with Quantum Leakage}
\label{sec:alt-extraction}

Alternating extraction is a technique for generating pseudo-random bits by repeatedly applying a seeded extractor to two independent weak sources in an alternating fashion. In each round, one source is used together with a public seed (often derived from the previous round) to extract fresh randomness, which then serves as the seed in the next round. This framework underlies various leakage-resilient constructions, and its behavior under leakage has been studied in the classical setting~\cite{DP08LeakageResilientStandard}.

In the quantum setting, however, new challenges arise. An adversary may hold entangled quantum side-information and interact with the system via general quantum channels that leak partial information at each round. These leakage operations may introduce correlations that are not captured by classical leakage models.

In this section, we extend alternating extraction to the setting of quantum leakage, using our framework of quantum computational unpredictability entropy. Building on our leakage chain rule, we show that unpredictability entropy degrades in a controlled way across rounds, and that pseudo-randomness can still be securely extracted, even under repeated quantum leakage. Our results generalize classical analysis and complement prior quantum work on computational entropy~\cite{CC17Computationalminentropy}.

We begin by formally defining a quantum variant of the ``Only Computation Leaks'' (OCL) model, and show that it preserves natural properties necessary for entropy evolution. We then analyze alternating extraction in this model and prove that security is maintained under computational bounds.

\subsection{New Leakage Model: Only Computation Quantum Leaks}\label{subsec:OCL}

Following the classical ``Only Computation Leaks'' (OCL) model~\cite{OCLmodel_04}, we consider a setting in which leakage occurs during active computation but not during storage or idle phases. We extend this framework to the quantum setting by allowing a bounded number of qubits to leak during each computational round.

We allow the leakage channel to depend on the adversary's current quantum state and to entangle the new leakage and existing side-information. The leakage process may vary from round to round and be chosen adaptively by the adversary.
This defines a more general class of leakage channels than those studied in the prior quantum work~\cite{CC17Computationalminentropy}.
This generalization reflects general quantum adversarial capabilities and is essential for analyzing multi-round protocols. In what follows, we formalize this model and discuss its implications for entropy degradation under quantum leakage. The components of the leakage channel are illustrated in~\cref{fig:bounded_leakage_channel}.

\begin{definition}\label{def:bounded-quantum-leakage-channel}
    Let $\rho_{XE}$ be a cq-state. A channel $\phi : \cS_{\circ}(XE) \to \cS_{\circ}(XLE')$ is called a $\lambda$-\emph{bounded quantum leakage channel} for $\rho_{XE}$ from if it can be written as a composition $\phi = \Lambda \circ \psi$, where:
    \begin{enumerate}
        \item $\psi : \cS_{\circ}(E) \to \cS_{\circ}(E')$ is a pre-processing CPTP map acting only on the adversary's system $E$ (it does not access~$X$),
        \item $\Lambda:\cS_{\circ}(XE') \to \cS_{\circ}(LXE')$ is a CPTP map that first appends an ancilla register $L$ in the state $\proj{0}_L$, of dimension at most $2^{\lambda}$, and then modifies only the $L$ register, leaving the $XE'$ marginal invariant:
              \begin{equation*}
                  \ptr{L}{\Lambda(\rho_{XE'})} = \rho_{XE'} \;.
              \end{equation*}
    \end{enumerate}
    The register $L$ is interpreted as the leaked information. The adversary holds the state in the registers $LE'$ after the application of $\phi$.
\end{definition}

\begin{figure}[ht]
    \centering
    \resizebox{0.5\textwidth}{!}{\begin{tikzpicture}[
    squarednode/.style={rectangle, draw=black!60, fill=black!5, very thick, minimum size=7mm},
    evenode/.style={rectangle, draw=red!60, fill=black!5, very thick, minimum size=7mm},
]

\node[squarednode] (X) {$\rho_{X}$};
\node[evenode] (E) [right=of X, xshift=2cm] {$\rho_{E}$};
\node[squarednode] (ER) [right=of E, xshift=-0.9cm] {$\rho_{R}$};

\node[evenode] (E') [below=of E] {$\rho_{E'}$};
\node[squarednode] (E'R) [below =of ER] {$\rho_{R'}$};

\node[evenode] (E'L) [below=of E'] {$\rho_{E'L}$};

\node[squarednode] (E'R2) [below =of E'R] {$\rho_{R'}$};

\node[squarednode] (X')[below= of X] {$\rho_{X}$};

\node[] (phi) [below=of X', xshift=1.7cm,yshift=0.7cm] {$\phi$};

\node[squarednode] (X'')[below= of X'] {$\rho_{X}$};

\draw[->] (X) -- (X');

\draw[->] (X') -- (phi);

\draw[->] (X') -- (X'');

\draw[->] (E) -- (E');

\draw[->] (E') -- (phi);

\draw[->] (ER) -- (E'R);

\draw[->] (phi) -- (E'L);

\end{tikzpicture}}
    \caption{During computation, the adversary leaks using a channel $\phi$. Only $L$ can carry information using this channel. All the other marginals remain unchanged. Between leakage rounds, the adversary's marginal can evolve independently.}
    \label{fig:bounded_leakage_channel}
\end{figure}

The definition ensures that only $L$ carries new information from $X$ to the adversary. The following lemma demonstrates an attack that uses leakage channels without the requirement for decomposition to one part acting only on $E$ separately from leaking $L$, that can ``leak'' all of $X$ to $E$ even if we set $\abs{L}=0$.
\begin{lemma}[The Invariance of $XE$ is Necessary]\label{overriding_E_attack}
    For any $k$, there exists $\rho_{XE}$ cq-state, and a channel $\psi : \cS_{\circ}(XE) \to \cS_{\circ}(XE')$, where:
    \begin{equation*}
        \left(\psi_{XE\to XE'}\right) (\rho_{XE}) = \sigma_{XE'} \;,
    \end{equation*}
    such that $H_{\min}(X|E)_{\rho} \ge k$ but $H_{\min}(X|E')_{\sigma} \le 0$.
\end{lemma}
\begin{proof}
    Let $\rho_{X}$ be maximally mixed on $k$ bits, and $E$ be a state of length $k$ in the all $0$ state.
    The adversary can `leak' using $\CNOT$ gates, controlled by $X$ on $E$. It's easy to see that $\sigma_{XE'}$ contains two identical copies of $X$, by measuring the state in the register $E'$ the adversary can guess the state at $X$ with probability~$1$, therefore
    \begin{equation*}
        H_{\min}(X|E)_{\rho} \ge k \;, \text{and} \quad H_{\min}(X|E')_{\sigma} \le 0 \;. \qedhere
    \end{equation*}
\end{proof}

Note that in this attack, there is no leakage register $L$ at all, and still, the adversary leaked all of the information out of $X$ by creating entanglement between the secret and the quantum side-information. Clearly, the channel $\psi$ does not decompose into (1) independent evolution of $E$ and (2) a leakage channel that can only leak using a fresh ancilla $L$, without modifying the marginal state on $XE$.

The condition that $\rho_{XE'}$ is invariant under the leakage channel can be viewed as a requirement that \emph{only} $L$ carries new information and correlations from $X$ to the adversary.
Any other information the adversary may gain comes from the state they already hold. The requirements do not, however, prevent leakage channels such as controlled operations on $L$ that are controlled by entries in $X$ or $E'$ or combinations of both. Such controlled operations may, for example, create entanglement between $L$ and $E'$.

It is often simpler to restrict the discussion to only unitary or isometric operations. This can generally be done if we additionally allow for an auxiliary system $R$, using Stinespring dilation theorem~\cite{stinespring1955dilationpositive}. The theorem essentially states that quantum channels cannot destroy information, only transfer information to the environment.
\begin{lemma}[Stinespring Dilation~\cite{stinespring1955dilationpositive}]\label{lem:StinespringDilatio}
    For any $\mathrm{CPTP}$ channel $\phi_{E\to E'}$ there is an auxiliary system $R$ and a isometry $\psi_{E\to E'R}$ such that for any $\rho_{E}$
    \begin{equation*}
        \phi(\rho) = \ptr{R}{\psi(\rho)} \;.
    \end{equation*}
\end{lemma}
We say that $\psi$ is the isometric version of $\phi$ with auxiliary system $R$.

\begin{lemma}\label{lem:smooth-min-entropy-leakage-chain-rule-bounded-channels}
    Let cq-state $\rho_{XE}$, $\varepsilon\ge 0$, and let $\phi$ be a $\lambda$-bounded quantum leakage channel, then
    \begin{equation*}
        H_{\min}^{\varepsilon}(X|LE')_{\phi(\rho)} \ge H_{\min}^{\varepsilon}(X|E)_{\rho} - 2\lambda \;.
    \end{equation*}
\end{lemma}

\begin{proof}
    From the definition~\cref{def:bounded-quantum-leakage-channel}, $\phi= \Lambda \circ  (\1_{X}\otimes\psi_{E\to E'})$ for some $\mathrm{CPTP}$ channel $\psi$. From the data-processing inequality for smooth min-entropy~\cite{Tomamichel_2016}, therefore
    \begin{equation*}
        H_{\min}^{\varepsilon}(X|E')_{\psi(\rho)} \ge H_{\min}^{\varepsilon}(X|E)_{\rho}\;.
    \end{equation*}
    By the leakage chain rule for smooth min-entropy~\cite{WTHR11Impossibilitygrowingquantumbitcommit}:
    \begin{equation*}
        H_{\min}^{\varepsilon}(X|LE')_{\phi(\rho)} + 2\lambda \ge H_{\min}^{\varepsilon}(X|E')_{\mathrm{Tr}_{L}\circ \Lambda \circ \psi(\rho)}= H_{\min}^{\varepsilon}(X|E')_{\psi(\rho)} \ge  H_{\min}^{\varepsilon}(X|E)_{\rho}\;. \qedhere
    \end{equation*}
\end{proof}

For an analogous statement about unpredictability entropy, we need an additional assumption, that the part of the channel that acts on the state of the adversary, denoted $\psi: \cS_{\circ}(E) \to \cS_{\circ}(E')$, needs to be implementable by a quantum circuit of size at most $t$. Note that the leakage part, the part of the channel that generates $L$, may still have unbounded complexity.

\begin{lemma}\label{lem:unp-entropy-leakage-chain-rule-bounded-channels}
    For any cq-state $\rho_{XE}$, any $\varepsilon \ge 0$ and $\phi = \Lambda \circ \psi$, a $\lambda$-bounded quantum leakage channel. Assuming the channel $\psi: \cS_{\circ}(E) \to \cS_{\circ}(E')$ can be implemented by a circuit of size $t$. For every $s$
    \begin{equation*}
        \hunps^{\varepsilon}(X|LE')_{\phi(\rho)} \ge \hunps[2s + 2\lambda + 5 + t]^{\varepsilon}(X|E)_{\rho} - 2\lambda \;.
    \end{equation*}
\end{lemma}

\begin{proof}
    Since $\psi_{E\to E'}$ is an isometry on the adversary side that can be implemented by a circuit of size $t$, by~\cref{lem:data_processing_inequality_unpredictability_entropy}, the data-processing inequality for unpredictability entropy we get
    \begin{equation*}
        \hunps^{\varepsilon}(X|E')_{\psi(\rho)} \ge \hunps[s+t]^{\varepsilon}(X|E)_{\rho} \;.
    \end{equation*}
    Using the fact that $\ptr{L}{\Lambda(\rho_{XE'})} = \rho_{XE'}$ and the leakage chain rule for unpredictability entropy in~\cref{lem:chain-rule-quantum-unpredictability-entropy}
    \begin{equation*}
        \hunps^{\varepsilon}(X|LE')_{\phi(\rho)} + 2\lambda \ge \hunps[2s+2\lambda + 5]^{\varepsilon}(X|E')_{\psi(\rho)} \ge  \hunps[2s+2\lambda+5 + t]^{\varepsilon}(X|E)_{\rho} \;. \qedhere
    \end{equation*}
\end{proof}

\subsection{Alternating Extraction}\label{subsec:alternating-extraction}
We now turn to an application of our leakage model: analyzing alternating extraction protocols in the presence of quantum leakage. Alternating extraction is a central primitive in leakage-resilient cryptography, in which two independent weak sources are used in alternating roles to extract fresh private randomness across multiple rounds. Our goal is to show that such protocols remain secure even when a bounded number of qubits leak to a quantum adversary after each computational step.

Our analysis builds on the model introduced in~\cref{subsec:OCL} and extends the classical framework of alternating extraction~\cite{DP08LeakageResilientStandard} to the setting of quantum leakage, as illustrated in~\cref{fig:alternatingextwithQuantumleakagenoPRG} below. We begin with the information-theoretic case, assuming min-entropy sources and general quantum-proof seeded extractors. This setting highlights the utility of our quantum leakage model and serves as a simpler stepping stone before addressing the more subtle case of computational entropy.

\begin{figure}[ht]
    \centering
    \resizebox{0.8\textwidth}{!}{\begin{tikzpicture}[
    roundnode/.style={circle, draw=black!60, fill=black!5, very thick, minimum size=7mm},
    squarednode/.style={rectangle, draw=black!60, fill=black!5, very thick, minimum size=7mm},
    evenode/.style={rectangle, draw=red!60, fill=black!5, very thick, minimum size=7mm},
]
\node[squarednode]      (A0)                            {$\rho_{A}$};
\node[]                 (A00)       [below = of A0]		{};
\node[squarednode]      (A1)       	[below = of A00] 	{$\rho_{A}$};
\node[]                 (A10)       [below = of A1]		{};
\node[squarednode]      (A2)       	[below = of A10] 	{$\rho_{A}$};
\node[]                 (A20)       [below = of A2]		{};
\node[squarednode]      (A3)       	[below = of A20] 	{$\rho_{A}$};
\node[]                 (A-0)       [right = of A0]		{};

\node[squarednode]      (K0)       	[right = of A-0] 	{$\rho_{K_{0}} $}; 
\node[]                 (K00)       [below = of K0]		{};
\node[squarednode]      (K1)       	[below = of K00] 	{$\rho_{K_1} = \rho_{\Ext(K_{0},B)}$}; 
\node[]                (K10)        [below = of K1]		{};
\node[squarednode]      (K2)       	[below = of K10] 	{$\rho_{K_2} = \rho_{\Ext(K_{1},A)}$}; 
\node[]                 (K20)       [below = of K2]		{};
\node[squarednode]      (K3)       	[below = of K20] 	{$\rho_{K_3} = \rho_{\Ext(K_{2},B)}$};

\node[]                 (B-0)       [right = of K0]		{};
\node[squarednode]      (B0)       	[right = of B-0] 	{$\rho_{B}$};
\node[]                 (B00)       [below = of B0]		{};
\node[squarednode]      (B1)       	[below = of B00] 	{$\rho_{B}$};
\node[]                 (B10)       [below = of B1]		{};
\node[squarednode]      (B2)       	[below = of B10] 	{$\rho_{B}$};
\node[]                 (B20)       [below = of B2]		{};
\node[squarednode]      (B3)       	[below = of B20] 	{$\rho_{B}$};
\node[]                 (B--Q)      [right = of B0]		{};
\node[]                 (B-Q)       [right = of B--Q]	{};
\node[evenode]          (Q0)       	[right = of B-Q] 	{$\rho_{E_{0}R_{0}}$};
\node[]                 (Q00)       [below = of Q0]		{};
\node[evenode]          (Q1)       	[below = of Q00] 	{$\rho_{E_{1}R_{1}}$};
\node[]                 (Q10)       [below = of Q1]		{};
\node[evenode]          (Q2)       	[below = of Q10] 	{$\rho_{E_{2}R_{2}}$};
\node[]                 (Q20)       [below = of Q2]		{};
\node[evenode]          (Q3)       	[below = of Q20] 	{$\rho_{E_{3}R_{3}}$};


\draw[->] (A0) to (A1);
\draw[->] (A1) to (A2); 
\draw[->] (A2) to (A3);
\draw[->] (B0) to (B1);
\draw[->] (B1) to (B2);
\draw[->] (B2) to (B3);

\draw[->] (K0) to (K1);
\draw[->] (K1) to (K2);
\draw[->] (K2) to (K3);
\draw[->] (B0) to (K1);
\draw[->] (A1) to (K2);
\draw[->] (B2) to (K3);


\path (B0) to node[] (f1) {$\phi_{1}(\rho_{BE_{0}R_{0}})$} (Q1);
\draw [red,->] (B0) to (f1);
\draw [red,->] (Q0) to (f1);
\draw [blue,->] (f1) to (Q1);

\path (B1) to node[] (f2) {$\phi_{2}(\rho_{AE_{1}R_{1}})$} (Q2);
\draw [red,->] (A1) to [out = -20] [in = 180]  (f2);
\draw [red,->] (Q1) to (f2);
\draw [blue,->] (f2) to (Q2);

\path (B2) to node[] (f3) {$\phi_{3}(\rho_{BE_{2}R_{2}})$} (Q3);
\draw [red,->] (B2) to (f3);
\draw [red,->] (Q2) to (f3);
\draw [blue,->] (f3) to (Q3);

\draw [green,->] (K0) to [out=-40] [in = 165] (Q1);
\draw [green,->] (K1) to [out=-40] [in = 165] (Q2);
\draw [green,->] (K2) to [out=-40] [in = 165] (Q3);


\end{tikzpicture}}
    \caption{Alternating extraction with quantum leakage. Black lines represent the alternating extraction steps without leakage. Red lines are inputs to the leakage channels, blue lines are the corresponding leakage outputs, and green lines indicate that the used seed becomes public after each extraction round.}
    \label{fig:alternatingextwithQuantumleakagenoPRG}
\end{figure}

Specifically, we consider two independent classical sources $X,Y$ of high min-entropy, along with a uniformly random seed $K_0$. The adversary initially holds some quantum side-information $E_0R_0$, and interacts with the system in each round via bounded-dimensional quantum leakage. We aim to show that throughout the protocol, (1) the independence of the sources conditioned on the adversary's state is preserved, and (2) the min-entropy of each source degrades in a controlled way as a function of the leakage dimension.

We begin by formalizing the notion of conditional independence in the quantum setting:

\begin{definition}\label{def:quantum_markov_chains}
    Let $\rho_{XYC}$ be a ccq-state. We say $X,Y$ are independent given $C$ if there exists a recovery map $\cR_{C \to CY}$ such that $\rho_{XCY} = \1_{X\to X}\otimes\cR_{C \to CY} (\rho_{XC})$.
\end{definition}
In~\cite{Hayden_2004_QuantumMarkovchainsStrongSubadditivity}, the authors show the following equivalency:
\begin{lemma}\label{lem:marov_chain_condition_Haydenrecovery}
    Let $\rho_{XYC}$ be a ccq-state. There exists a recovery map $\rho_{XCY} = \1_{X\to X}\otimes\cR_{C \to CY} (\rho_{XC})$ if and only if
    $$I(X : Y | C)_{\rho} = H(X|C)_{\rho}-H(X|YC)_{\rho}= 0\;.$$
\end{lemma}
We say that such states satisfy the Markov Chain condition, denoted by $X \mch C \mch Y$. We can see that one side of this equivalence can be easily proven using only the data-processing inequality for conditional entropy.

\begin{proof}
    Let $\rho_{XYC}$ be a ccq-state such that there is a map $\rho_{XCY} = (\1_{X\to X}\otimes\cR_{C \to CY} )(\rho_{XC})$. By the data prepossessing inequality of two local maps $\cR_{C \to CY}$ and $\Tr_{Y} : \cS_{\bullet}(YC) \to \cS_{\bullet}(C)$ we see that:
    \begin{equation*}
        H(X|C)_{\rho} \le H(X|YC)_{(\1_{X\to X} \otimes \cR_{C\to CY}) (\rho)} = H(X|YC)_{\rho} \le H(X|C)_{\rho} \;.
    \end{equation*}
    Since $H(X|C)_{\rho} = H(X|C)_{\rho}$ we can conclude that
    $I(X : Y | C)_{\rho} = H(X|C)_{\rho}-H(X|YC)_{\rho}= 0$. %
\end{proof}
Noticing that the proof requires only that the conditional entropy measure we use satisfies the data processing inequality, we can conclude a more general statement about states that satisfy the Markov chain condition.
\begin{corollary}
    For any ccq-state $\rho_{XYC}$, $\varepsilon\ge0$:
    $$I(X : Y | C)_{\rho} = 0 \implies H_{\min}^{\varepsilon}(X|C)_{\rho}-H_{\min}^{\varepsilon}(X|YC)_{\rho}= 0\;.$$
\end{corollary}
Moreover, the corollary holds for any conditional entropy measure that satisfies the data-possessing inequality.

Let us formally define the process of alternating extraction under bounded quantum leakage in the Only Computation Quantum Leaks model
The process of alternating extraction under bounded quantum leakage in the Only Computation Quantum Leaks model is formally specified in~\cref{fig:Alternating_ext}.

\begin{figure}[ht]
    \begin{tcolorbox}[colback=white,colframe=blue!50, title = Alternating extraction under quantum leakage]
        Let $\rho_{XYE_0R_0}$ be a cq-state, classical on $XY$, such that $I(X : Y | E_0R_0)_\rho = 0$. Let $\rho_{K_0}$ be independent and uniformly random.

        Initialize $i = 0$ and repeat the following steps:
        \begin{enumerate}
            \item \label{i-step1_altext} If $i$ is even, set $T_{i} = Y$; otherwise, set $T_{i} = X$.
                  \hfill \textit{[Choose active source]}

            \item Compute $\rho_{K_{i+1}} = \rho_{\Ext(K_i, T_{i})}$.
                  \hfill \textit{[Extract using current source and seed]}

            \item \label{altext_dilation_isometry_step} Let $\psi_i$ be a $\lambda$-bounded quantum leakage channel for $\rho_{T_{i}E_{i}R_{i}}$, denote $\rho_{T_{i}LE'_{i}R'_{i}}=\psi_{i}(\rho_{T_{i}E_{i}R_{i}})$. \hfill \textit{[Leakage]}

            \item The state, including the auxiliary register, after leakage: \hfill \textit{[Updating registers]}
                  \begin{equation*}
                      \rho'_{T_{i}E_{i+1}R_{i+1}} = \phi_i(\rho_{T_{i}E_iR_i}), \quad
                      E_{i+1} = E'_iL_i \;.
                  \end{equation*}

            \item Set the final adversary state to include the public seed: \hfill \textit{[Seed becomes public]}
                  \begin{equation*}
                      \rho_{E_{i+1}R_{i+1}} := \rho'_{E_{i+1}R_{i+1}} \otimes \rho_{K_i}  \;.
                  \end{equation*}

            \item Increment $i$ and go to~\cref{i-step1_altext}.
        \end{enumerate}
    \end{tcolorbox}
    \caption{Alternating extraction protocol under quantum leakage}
    \label{fig:Alternating_ext}
\end{figure}

First, we show that the quantum leakage model preserves the quantum Markov chain condition. That is, assuming that~$X,Y$ are independent, conditioned on the state of the adversary and the environment, we show they remain independent under alternating bounded quantum leakage channels. This result can be viewed as a quantum analog of~\cite[Lemma 2]{DP08LeakageResilientStandard}, which shows that classical leakage from one source in the classical OCL model preserves conditional independence.
In the protocol described in~\cref{fig:Alternating_ext}, we show that $\lambda$-bounded quantum leakage applied to one side of a tripartite system maintains the quantum Markov chain structure.
We split the proof into the two stages of the leakage channel. First, in~\cref{lemma:CPTP_on_E_preserves_markov} we show that isometric operations on the state of the adversary and the environment preserve the Markov chain structure. Following that, in~\cref{lemma:markov_preservation_leakage_L} we show that leakage from one source that only modifies a new register $L$, preserves the quantum Markov chain structure. Note that the leakage map described in~\cref{altext_dilation_isometry_step} of the protocol in~\cref{fig:Alternating_ext} decomposes into an isometry on $E_{i}R_{i}$ and a leakage step that only modifies~$L_{i}$.

\begin{lemma}[Markov Chains with Isometries on $ER$]
    \label{lemma:CPTP_on_E_preserves_markov}
    Let $\rho_{XYE_iR_i}$ be a ccqq-state such that
    \begin{equation*}
        I(X : Y | E_i R_i)_{\rho} = 0\;.
    \end{equation*}
    Let $\phi : \cS_{\circ}(E_iR_{i}) \to \cS_{\circ}(E'_{i}R'_{i})$ be an isometry. Denote,
    $\rho_{XYE'_i R'_i} := (\1_{XY} \otimes \phi_{E_{i}R_{i}\to E'_{i}R'_{i}})(\rho_{XYE_iR_i})$, then:
    \begin{equation*}
        I(X : Y | E'_i R'_i)_{\rho} = 0 \;.
    \end{equation*}
\end{lemma}

\begin{proof}
    Conditional mutual information is invariant by local isometries on the conditioning registers, as the von Neumann entropy itself is isometry-invariant (see \cite[Chapter~11]{nielsen2010quantumbook}). For completeness and clarity, we include here a proof constructing directly a recovery map. 
    
    By definition, the isometry $\phi$ has an inverse $\phi^\dagger$ such that
    \begin{equation*}
        (\1_{XY} \otimes \phi^\dagger_{E'_{i}R'_{i}\to E_{i}R_{i}})\rho_{XYE'_i R'_i} = \rho_{XYE_iR_i} \;.
    \end{equation*}
    Since $I(X : Y | E_i R_i)_{\rho} = 0$ we know there is a recovery map $\cR_{E_{i}R_{i}\to E_{i}R_{i}Y}$ such that
    \begin{equation*}
        (\1_{X}\otimes\cR_{E_{i}R_{i}\to E_{i}R_{i}Y})(\rho_{XE_iR_i}) = \rho_{XYE_iR_i}\;.
    \end{equation*}
    By composing the isometries with the recovery map we see that:
    \begin{equation*}
        \left(\1_{X}\otimes \left(\1_{Y}\otimes\phi_{E_{i}R_{i}\to E'_{i}R'_{i}}\right) \circ \cR_{E_{i}R_{i}\to E_{i}R_{i}Y} \circ \phi^\dagger_{E'_{i}R'_{i}\to E_{i}R_{i}}\right)(\rho_{XE'_iR'_i}) = \rho_{XYE'_iR'_i}\;.
    \end{equation*}
    Therefore, there is a recovery map and from~\cref{lem:marov_chain_condition_Haydenrecovery} we conclude:
    \begin{equation*}
        I(X : Y | E'_i R'_i)_{\rho} = 0 \;. \qedhere
    \end{equation*}
\end{proof}

\begin{lemma}[Markov Chains and Leakage from One Source]
    \label{lemma:markov_preservation_leakage_L}
    Let $\rho_{XYE'_iR_i}$ be a ccqq-state such that $I(X : Y | E'_i R_i)_{\rho} = 0$. Let $\Lambda : \cS_{\circ}(YE'_iR_i) \to \cS_{\circ}(YE'_iLR_{i})$ be a CPTP map such that,
    \begin{equation*}
        \Tr_L[\Lambda(\rho_{YE'_iR_{i}})] = \rho_{YE'_iR_{i}} \;.
    \end{equation*}
    Define
    \begin{equation*}
        \rho_{XYE_{i+1} R_{i+1}} := (\1_X \otimes \Lambda_{YE'_{i}} \otimes \1_{R_{i}})(\rho_{XYE'_i R_i}) \;,
    \end{equation*}
    where $E_{i+1} := E'_i L$ and $R_{i+1}=R_{i}$. Then,
    \begin{equation*}
        I(X : Y | E_{i+1} R_{i+1})_{\rho} = 0 \;.
    \end{equation*}
\end{lemma}

\begin{proof}
    Since $I(X : Y | E'_i R_i) = 0$, there exists a recovery map $\mathcal{R}_{E'_i R_i \to E'_i R_i Y}$ such that
    \begin{equation*}
        \rho_{XE'_iR_iY} = (\1_X \otimes \mathcal{R})(\rho_{XE'_iR_i}) \;.
    \end{equation*}
    We define a new recovery map for $E_{i+1} R_i = E'_i L R_{i+1}$ by discarding $L$ and composing $\mathcal{R}$ with $\Lambda$:
    \begin{equation*}
        \widetilde{\mathcal{R}} := \Lambda \circ \mathcal{R} \circ \Tr_L \;.
    \end{equation*}
    By discarding $L$, we return to a state that we know satisfies the Markov chain condition. From~\cref{lem:marov_chain_condition_Haydenrecovery} we know there is a recovery map to reconstruct $Y$ for this marginal state. We can apply the leakage channel again on the recovered state to return to the full state after leakage.
    Note that this map does take $Y$ as an input, as required from recovery maps for quantum Markov chains. The new recovery map we get from this composition of maps is:
    \begin{equation*}
        \rho_{XYE_{i+1}R_{i+1}} = (\1_X \otimes \widetilde{\mathcal{R}})(\rho_{XE_{i+1}R_{i+1}})\;,
    \end{equation*}
    showing that $I(X : Y | E_{i+1} R_{i+1}) = 0$. %
\end{proof}

We can also see this from an entropic point of view. By the data prepossessing inequality for the above channels, we get the following sequence of intentionalities:
\begin{align*}
    H(Y|E_{i}R_{i})_{\rho}
     & \le H(Y|XE_{i}R_{i})_{\cR(\rho)}                                                         \\
     & \le H(Y|XE_{i}R_{i}L)_{\Lambda (\cR(\rho))}                                              \\
     & \le H(Y|E_{i}R_{i}L)_{\Tr_{X}(\Lambda (\cR(\rho)))}                                      \\
     & \le H(Y|E_{i}R_{i})_{\Tr_{L}(\Tr_{X}(\Lambda (\cR(\rho))))}  = H(Y|E_{i}R_{i})_{\rho}\;.
\end{align*}
From this we can see that all the conditional entropies above are equal. From $H(Y|XE_{i}R_{i}L)_{\rho}=H(Y|E_{i}R_{i}L)_{\rho}$ we can conclude $I(X : Y | E_{i+1} R_{i+1}) = 0$.

Combining~\cref{lemma:CPTP_on_E_preserves_markov} and~\cref{lemma:markov_preservation_leakage_L}, by the definition of bounded quantum leakage channels~\cref{def:bounded-quantum-leakage-channel} we can see that bounded quantum leakage channels preserve quantum Markov chains. We state this in the following lemma for leakage from $Y$, the case for $X$ is symmetric.
\begin{lemma}\label{lemma:markov-quantum-from-alternating-with-quantum}
    Let $\rho_{XYE_{i}R_{i}}$ be a ccqq state such that
    \begin{equation*}
        I(X:Y|E_{i}R_{i})_{\rho} = 0\;.
    \end{equation*}
    Let $\rho_{XYE_{i+1}R_{i+1}}$ be the state of the system, including the environment, after the application of the isometry $\phi_{i}$
    the isometric version of $\psi_{i}$ a $\lambda$-bounded quantum leakage from $Y$ to $E_{i}R_{i}$, as described inin~\cref{fig:Alternating_ext}.
    \begin{equation*}
        I(X : Y | E_{i+1}R_{i+1})_{\rho} = 0\;.
    \end{equation*}
\end{lemma}
From the leakage chain rule for smooth min-entropy~\cref{lem:smooth-min-entropy-leakage-chain-rule-bounded-channels} we also know that, for even $i$
\begin{equation*}
    H_{\min}^{\varepsilon}(Y|E_i R_{i} )_{\rho} \ge H_{\min}^{\varepsilon}(Y|E_{i+1} R_{i+1})_{\rho} - 2\lambda \;,
\end{equation*}
and from~\cref{lemma:markov-quantum-from-alternating-with-quantum} for even $i$ we get
\begin{equation*}
    H_{\min}^{\varepsilon}(X|E_{i}R_{i})_{\rho} = H_{\min}^{\varepsilon}(X|E_{i+1}R_{i+1})_{\rho} \;.
\end{equation*}
For odd $i$, the roles of $X$ and $Y$ are switched.

We now show that randomness can still be securely extracted even when the seed is not perfectly uniform. The following lemma quantifies the degradation in extractor output quality when the seed is only close to uniform and the source has bounded leakage.
\begin{lemma}\label{lem:ext_with_epsilons}
    Let $\dist(\rho_{K_{i}E_{i}} , \rho_{U^{m}}\otimes\rho_{E_{i}}) \le \varepsilon$ and
    \begin{equation*}
        H_{\min}^{\varepsilon'}(Y|E_{i}R_{i}) \ge k_{\ext} + 2 \lambda\;.
    \end{equation*}
    Let $\Ext:\bits^{n} \times \bits^{d} \to \bits^{m}$ be a quantum proof seeded $k_{\ext},\varepsilon_{\ext}$ extractor, then:
    \begin{equation*}
        \dist(\rho_{\Ext(K_{i},Y)E_{i+1}R_{i+1}},\rho_{U^{m}}\otimes\rho_{E_{i+1}R_{i+1}}) \le 2(\varepsilon+\varepsilon') +\varepsilon_{\ext}\;.
    \end{equation*}
\end{lemma}

\begin{proof}
    The proof follows directly from the definition of quantum-proof seeded extractor along with the leakage chain rule for smooth min-entropy and~\cref{lemma:markov-quantum-from-alternating-with-quantum}.
    From the triangle inequality, we get the hybrid argument:
    \begin{align*}
        \dist(\rho_{\Ext(K_{i},Y)E_{i+1}},\rho_{U^{m}}\otimes\rho_{E_{i}})                         %
        \le & \dist(\rho_{\Ext(K_{i},Y)E_{i+1}},\tilde{\rho}_{\Ext(K_{i},Y)E_{i+1}}) +             %
        \dist(\tilde{\rho}_{\Ext(K_{i},Y)E_{i+1}},\rho_{U^{m}}\otimes\rho_{E_{i}})                 \\
        \le & 2\varepsilon +\dist(\rho_{K_{i}E_{i}},\rho_{U^{m}}\otimes \rho_{E_{i}})+             %
        \dist(\tilde{\rho}_{\Ext(K_{i},Y)E_{i+1}},\tilde{\rho}_{U^{m}}\otimes\tilde{\rho}_{E_{i}}) \\
        \le & 2\varepsilon+\varepsilon' +\varepsilon_{\ext} \qedhere
    \end{align*}
\end{proof}

From here, we can conclude that for any $i$, if there is sufficient min-entropy at the start, then $K_{i}$ is close to uniformly random and independent of the state of the adversary and the environment after leakage $E_{i+1}R{i+1}$. 

The following lemma summarizes the behavior of min-entropy and conditional independence throughout the alternating extraction process under quantum leakage. It shows that the entropy of the sources degrades in a controlled way with each leakage step, while the quantum Markov condition is preserved. This is the quantum analog of~\cite[Lemma 1]{DP08LeakageResilientStandard}, which analyzes alternating extraction under classical leakage. Our result extends this to the quantum setting using smooth min-entropy and our leakage model.

\begin{lemma}[Alternating Extraction]\label{lem:alternating_extraction}
    Let $\rho_{XYE_{0}R_{0}K_{0}}$ be defined as in the beginning of the protocol in~\cref{fig:Alternating_ext}, and
    \begin{equation*}
        H_{\min}^{\varepsilon}(X|E_0 R_{0})_{\rho} \ge k \quad H_{\min}^{\varepsilon}(Y|E_0 R_{0})_{\rho} \ge k \;.
    \end{equation*}
    For every $i$ in the alternating extraction protocol
    \begin{equation}\label{eq:mch_in_alternating}
        X \mch E_{i}R_{i} \mch Y \;.
    \end{equation}
    \begin{align}\label{eq:leakage_in_alternating}
        H_{\min}^{\varepsilon}(X|E_{i}R_{i})_{\rho} & \ge k - (1+(-1)^{i+1} + 2i) \lambda             \\
        H_{\min}^{\varepsilon}(Y|E_{i}R_{i})_{\rho} & \ge k - (1+(-1)^{i} + 2i) \lambda \;. \nonumber
    \end{align}
\end{lemma}

\begin{proof}
    The proof follows directly repeated use of~\cref{lem:smooth-min-entropy-leakage-chain-rule-bounded-channels},~\cref{lemma:markov-quantum-from-alternating-with-quantum} and~\cref{lem:ext_with_epsilons} in secession, $i$ times, alternating between $X$ and $Y$ as the min-entropy source. %
\end{proof}

\begin{lemma}
    For every $i$ such that $k - (1+(-1)^{i} + 2i) \lambda > k_{\ext}$,
    \begin{equation*}
        \dist(\rho_{\Ext(K_{i},Y)E_{i+1}},\rho_{U^{m}}\otimes\rho_{E_{i}}) \le i(2\varepsilon + \varepsilon_{\ext})\;.
    \end{equation*}
\end{lemma}

\begin{proof}
    The proof follows from the bounds on min-entropy from~\cref{lem:alternating_extraction} and the security of quantum-proof seeded extractors~\cref{def:quantum_proof_seeded_ext}. %
\end{proof}

\subsection{Alternating Extraction from Unpredictability Entropy}
\label{subsec:alt-extraction-unp}

We now extend our analysis of alternating extraction to the setting where the sources possess high quantum \emph{computational} unpredictability entropy. In contrast to the information-theoretic case analyzed in the previous subsection, where entropy is measured against unbounded adversaries, we now consider computationally bounded adversaries, and track how unpredictability evolves under repeated quantum leakage.

A key difference in this setting is that we do not use the output of one extraction round as the seed for the next. Instead, we assume that each round is initialized with a fresh, uniformly random public seed, independent of the adversary's state. This modification is necessary because the extractors we analyze in this setting do not support output lengths longer than the seed length, making it unsuitable to recycle extractor outputs as future seeds.

As before, after each round of extraction, a bounded number of qubits may leak to the adversary through a quantum leakage channel. Our goal is to show that unpredictability entropy degrades in a controlled fashion across rounds, and that fresh pseudo-random bits can still be securely extracted, provided the initial unpredictability is sufficiently high and the leakage dimension per round is bounded.

To this end, we combine our leakage chain rule for computational unpredictability entropy with an inductive analysis of entropy degradation over rounds. This allows us to extend the alternating extraction framework to the computational setting, under general quantum leakage.

\begin{figure}[!hb]
    \begin{tcolorbox}[colback=white,colframe=blue!50, title = Alternating extraction under quantum leakage with unpredictability sources and fresh seeds.]
        Let $\rho_{XYE_0R_0}$ be a ccqq-state, classical on $XY$, such that $I(X : Y | E_0R_0)_\rho = 0$.

        Initialize $i = 0$ and repeat the following steps:
        \begin{enumerate}
            \item \label{i-step1_altextunp} If $i$ is even, set $T = Y$, otherwise, set $T = X$.
                  \hfill \textit{[Choose active source]}

            \item Sample a fresh, uniform, public seed $\rho_{S_i}$.
                  \hfill \textit{[Seed is fresh and independent]}

            \item Compute the extractor output: \hfill \textit{[Extracting private pseudo-randomness]}
                  \begin{equation*}
                      \rho_{K_i} = \rho_{\Ext(T_{i}, S_i)} \;.
                  \end{equation*}
            \item Let $\varphi_i$ be a $\lambda$-bounded quantum leakage channel with circuit size $t$ for $\rho_{T_{i}E_{i}}$, such that $\varphi_i = \phi_i \circ \psi_i$, where:
                  \begin{itemize}
                      \item $\psi_i:E_iR_{i} \to E'_{i}R'_{i}$ is an isometry.
                      \item $\phi_i$ is a CPTP map acting on $(T_{i}, E'_i)$ that appends a leakage register $L_i$ of dimension at most $2^\lambda$, modifies only $L_i$, and preserves the marginal on $T_{i} E'_i$: \hfill \textit{[Bounded leakage]}
                            \begin{equation*}
                                \ptr{L_i}{\phi_i(\rho_{T_{i} E'_i R_i})} = \rho_{T_{i} E'_i R_i} \;.
                            \end{equation*}
                  \end{itemize}
                  The adversary's state, including the auxiliary system, after the leakage in round $i$ is:
                  \begin{equation*}
                      \rho_{E_{i+1} R_{i+1}} := \phi_i(\rho_{T_{i} E'_i R'_i}) \;, \quad \text{where} \quad E_{i+1} := E'_i L_i \;.
                  \end{equation*}

            \item The final state after round $i$ is: \hfill \textit{[Seeds are public and independent]}
                  \begin{equation*}
                      \rho_{XYE_{i+1}R_{i+1}} := \phi_i\circ \psi_i(\rho_{XYE_{i}R_{i}}) \otimes \rho_{S_i} \;.
                  \end{equation*}

            \item Increment $i$ by 1 and go to \cref{i-step1_altextunp}.
        \end{enumerate}
    \end{tcolorbox}
    \caption{Alternating extraction protocol under quantum leakage with unpredictability sources and fresh seeds}
    \label{fig:Alternating_ext_unp_fresh_seeds}
\end{figure}

\begin{figure}[ht]
    \centering
    \resizebox{0.8\textwidth}{!}{\begin{tikzpicture}[
        roundnode/.style={circle, draw=black!60, fill=black!5, very thick, minimum size=7mm},
        squarednode/.style={rectangle, draw=black!60, fill=black!5, very thick, minimum size=7mm},
        evenode/.style={rectangle, draw=red!60, fill=black!5, very thick, minimum size=7mm},
    ]
    \node[squarednode]      (A0)                            {$\rho_{A}$};
    \node[]                 (A00)       [below = of A0]		{};
    \node[squarednode]      (A1)       	[below = of A00] 	{$\rho_{A}$};
    \node[]                 (A10)       [below = of A1]		{};
    \node[squarednode]      (A2)       	[below = of A10] 	{$\rho_{A}$};
    \node[]                 (A20)       [below = of A2]		{};
    \node[squarednode]      (A3)       	[below = of A20] 	{$\rho_{A}$};
    \node[]                 (A-0)       [right = of A0]		{};

    \node[]                 (S0)       	[right = of A-0] 	{}; 
    \node[squarednode]      (S00)       [below = of S0]		{$\rho_{S_{0}} $};
    \node[squarednode]      (K1)       	[right = of A1] 	{$\rho_{K_1} = \rho_{\Ext(S_{0},B)}$}; 
    \node[squarednode]      (S10)       [below = of K1]		{$\rho_{S_{1}} $};
    \node[squarednode]      (K2)       	[right = of A2] 	{$\rho_{K_2} = \rho_{\Ext(S_{1},A)}$}; 
    \node[squarednode]      (S20)       [below = of K2]		{$\rho_{S_{2}} $};
    \node[squarednode]      (K3)       	[right = of A3] 	{$\rho_{K_3} = \rho_{\Ext(S_{2},B)}$};

    \node[]                 (B-0)       [right = of S0]		{};
    \node[squarednode]      (B0)       	[right = of B-0] 	{$\rho_{B}$};
    \node[]                 (B00)       [below = of B0]		{};
    \node[squarednode]      (B1)       	[below = of B00] 	{$\rho_{B}$};
    \node[]                 (B10)       [below = of B1]		{};
    \node[squarednode]      (B2)       	[below = of B10] 	{$\rho_{B}$};
    \node[]                 (B20)       [below = of B2]		{};
    \node[squarednode]      (B3)       	[below = of B20] 	{$\rho_{B}$};
    \node[]                 (B--Q)      [right = of B0]		{};
    \node[]                 (B-Q)       [right = of B--Q]	{};
    \node[evenode]          (Q0)       	[right = of B-Q] 	{$\rho_{E_{0}R_{0}}$};
    \node[]                 (Q00)       [below = of Q0]		{};
    \node[evenode]          (Q1)       	[below = of Q00] 	{$\rho_{E_{1}R_{1}}$};
    \node[]                 (Q10)       [below = of Q1]		{};
    \node[evenode]          (Q2)       	[below = of Q10] 	{$\rho_{E_{2}R_{2}}$};
    \node[]                 (Q20)       [below = of Q2]		{};
    \node[evenode]          (Q3)       	[below = of Q20] 	{$\rho_{E_{3}R_{3}}$};


    \draw[->] (A0) to (A1);
    \draw[->] (A1) to (A2); 
    \draw[->] (A2) to (A3);
    \draw[->] (B0) to (B1);
    \draw[->] (B1) to (B2);
    \draw[->] (B2) to (B3);

    \draw[->] (S00) to (K1);
    \draw[->] (S10) to (K2);
    \draw[->] (S20) to (K3);

    \draw[->] (B0) to (K1);
    \draw[->] (A1) to (K2);
    \draw[->] (B2) to (K3);


    \path (B0) to node[] (f1) {$\phi_{1}(\rho_{BE_{0}R_{0}})$} (Q1);
    \draw [red,->] (B0) to (f1);
    \draw [red,->] (Q0) to (f1);
    \draw [blue,->] (f1) to (Q1);

    \path (B1) to node[] (f2) {$\phi_{2}(\rho_{AE_{1}R_{1}})$} (Q2);
    \draw [red,->] (A1) to [out = -20] [in = 180]  (f2);
    \draw [red,->] (Q1) to (f2);
    \draw [blue,->] (f2) to (Q2);

    \path (B2) to node[] (f3) {$\phi_{3}(\rho_{BE_{2}R_{2}})$} (Q3);
    \draw [red,->] (B2) to (f3);
    \draw [red,->] (Q2) to (f3);
    \draw [blue,->] (f3) to (Q3);



\end{tikzpicture}}
    \caption{Alternating extraction with quantum leakage and fresh seeds. Black lines represent the alternating extraction steps without leakage. Red lines are inputs to the leakage channels, blue lines are the corresponding leakage outputs, and green lines indicate that the seeds are public.}
    \label{fig:alternatingextwithQuantumleakagenofreshhseed}
\end{figure}

The Markov chain condition $X \mch E_i R_i \mch Y$ is preserved throughout the execution of the protocol. This follows directly from the same argument as in~\cref{lemma:markov-quantum-from-alternating-with-quantum}, since the leakage model and its decomposition are unchanged. The use of fresh public seeds does not affect this structure.

\begin{lemma}[Unpredictability Entropy After Round $i$]
    \label{lem:unp_entropy_after_round}
    Let $\rho_{XYE_0R_0}$ be a ccqq-state, classical on $XY$, and suppose that
    \begin{equation*}
        I(X : Y | E_0 R_0)_\rho = 0 \;, \quad
        \hunps[s]^{\varepsilon}(X | E_0 R_{0} )_\rho \ge k \;, \quad
        \hunps[s]^{\varepsilon}(Y | E_0 R_{0})_\rho \ge k \;.
    \end{equation*}
    Assume the alternating extraction protocol of~\cref{fig:Alternating_ext_unp_fresh_seeds} is run for $i$ rounds,
    with fresh uniform seeds, $\lambda$-bounded leakage in each round, and adversary updates via circuits of size at most $t$.

    Then, for all $i \ge 0$, the unpredictability entropy satisfies:
    \begin{align*}
        \hunps[s - \bO(i(t + \lambda))]^{\varepsilon}(X | E_i R_{i})_\rho & \ge k - \delta_i^X \cdot 2\lambda \;, \\
        \hunps[s - \bO(i(t + \lambda))]^{\varepsilon}(Y | E_i R_{i})_\rho & \ge k - \delta_i^Y \cdot 2\lambda \;,
    \end{align*}
    where $\delta_i^X$ (resp. $\delta_i^Y$) denotes the number of rounds up to step $i$ in which $X$ (resp. $Y$) was used as the active source.
\end{lemma}

\begin{proof}
    We proceed by induction on $i$.

    The assumptions of the lemma directly give for $i=0$
    \begin{equation*}
        \hunps[s]^{\varepsilon}(X | E_0 R_0)_\rho \ge k \;, \quad
        \hunps[s]^{\varepsilon}(Y | E_0 R_0)_\rho \ge k \;.
    \end{equation*}

    Assume the claim holds for round $i$. We show it holds for round $i+1$.
    Let $T$ be the active source used in round $i$. Without loss of generality, suppose $T = X$ (the case $T = Y$ is symmetric). By the protocol definition, the adversary state is updated as
    \begin{equation*}
        \rho_{E_{i+1} R_{i+1}} := \phi_i(\rho_{T E'_i R_i}) \otimes \rho_{S_i} \;,
    \end{equation*}
    where $E'_i$ is obtained from $(E_i, R_i)$ via a size-$t$ circuit $\psi_i$, and $\phi_i$ is a $\lambda$-bounded quantum leakage channel acting on $T$ and $(E'_i, R_i)$.

    By the data-processing inequality~\cref{lem:data_processing_inequality_unpredictability_entropy}, the transformation $\psi_i$ degrades entropy by at most $t$ gates:
    \begin{equation*}
        \hunps[s - \bO(i(t + \lambda)) - t]^{\varepsilon}(X | E'_iR'_i)_\rho \ge \hunps[s- \bO(i(t + \lambda))]^{\varepsilon}(X | E_i R_i)_\rho \;.
    \end{equation*}
    Then, by the leakage chain rule for unpredictability entropy~\cref{lem:chain-rule-quantum-unpredictability-entropy}, the leakage channel $\phi_i$ reduces entropy by at most $2\lambda$ and adds an additional $\bO(\lambda)$ gate overhead:
    \begin{equation*}
        \hunps[s - \bO(i(t + \lambda)) - \bO(\lambda)]^{\varepsilon}(X | E_{i+1} R_{i+1})_\rho \ge \hunps[s- \bO(i(t + \lambda))]^{\varepsilon}(X | E'_i R'_{i})_\rho - 2\lambda \;.
    \end{equation*}
    Combining the two steps:
    \begin{equation*}
        \hunps[s - \bO((i+1)(t + \lambda))]^{\varepsilon}(X | E_{i+1} R_{i+1})_\rho \ge \hunps[s- \bO(i(t + \lambda))]^{\varepsilon}(X | E_iR_i)_\rho - 2\lambda \;.
    \end{equation*}
    By the induction hypothesis,
    \begin{equation*}
        \hunps[s - \bO(i(t+\lambda))]^{\varepsilon}(X | E_iR_i)_\rho \ge k - \delta_i^X \cdot 2\lambda \;,
    \end{equation*}
    and since $X$ was used in round $i$, we have $\delta_{i+1}^X = \delta_i^X + 1$, hence:
    \begin{equation*}
        \hunps[s - \bO((i+1)(t + \lambda))]^{\varepsilon}(X | E_{i+1} R_{i+1})_\rho \ge k - \delta_{i+1}^X \cdot 2\lambda \;.
    \end{equation*}

    For the passive source $Y$, note that it is untouched in round $i$, and the only transformation to the adversary state is via a size-$t$ circuit followed by a leakage map that does not act on $Y$. Therefore, $\delta_{i+1}^Y=\delta_{i}^Y$, applying the data-processing inequality:
    \begin{equation*}
        \hunps[s - \bO((i+1)(t + \lambda))]^{\varepsilon}(Y | E_{i+1} R_{i+1})_\rho \ge \hunps[s - \bO(i(t+\lambda))]^{\varepsilon}(Y | E_iR_i)_\rho \ge k - \delta_i^Y \cdot 2\lambda = k - \delta_{i+1}^Y \cdot 2\lambda \;,
    \end{equation*}
    since $Y$ was not active in this round. The gate overhead again accumulates linearly in $i(t+\lambda)$, and the claim holds for $i+1$. %
\end{proof}

Having established that unpredictability entropy degrades in a controlled manner across rounds, we now show that extraction can proceed whenever the active source retains sufficient entropy. The following lemma applies the extractor security guarantee to derive computational pseudo-randomness from the active source in each round.

\begin{lemma}[Extraction from Unpredictability Entropy with Fresh Seeds]\label{lem:ext_unp_fresh_seed}
    Let $T \in \set{X,Y}$ be the active source used in round $i$ of the protocol of~\cref{fig:Alternating_ext_unp_fresh_seeds}, and suppose:
    \begin{itemize}
        \item $\rho_{T E_iR_i}$ is a cq-state at the beginning of round $i$.
        \item $S_i$ is a fresh, uniform, public seed, independent of $T$ and $E_i$.
        \item $\Ext : \{0,1\}^n \times \{0,1\}^d \to \{0,1\}^m$ is a seeded extractor secure against quantum unpredictability entropy in the sense of~\cref{def:ext_pseudorandomness_from_unp}.
        \item The source satisfies $\hunps[s']^{\varepsilon'}(T | E_iR_i)_\rho \ge k_{\mathrm{ext}} + 2\lambda$.
    \end{itemize}
    Then, after applying the leakage channel $\varphi_i$ in round $i$, the output of $\Ext(T, S_i)$ is pseudorandom with respect to the adversary:
    \begin{equation*}
        \dist_{s}(\rho_{\Ext(T, S_i) S_i E_{i+1}R_{i+1}}, \rho_{U_m} \otimes \rho_{S_i E_{i+1}R_{i+1}}) \le \varepsilon_{\mathrm{ext}} + 2\varepsilon' \;,
    \end{equation*}
    for some $s = \bO(s')$ depending on the extractor construction.
\end{lemma}

\begin{proof}
    Since $S_i$ is uniform and independent of $T$ and $E_i$, we may apply the extractor guarantee for unpredictability entropy, such as~\cref{lem:m-bit-ext-from-unp-IP}, to conclude that
    \begin{equation*}
        \dist_{s}(\rho_{\Ext(T, S_i) S_i E_iR_i}, \rho_{U_m} \otimes \rho_{S_i E_iR_i}) \le \varepsilon_{\mathrm{ext}} + 2\varepsilon' \;.
    \end{equation*}
    The adversary state is updated by the leakage channel $\varphi_i$, which acts only on $T$ and $(E_i, R_i)$ and does not touch the seed $S_i$. In particular, this transformation can be absorbed into the distinguisher, increasing its size by at most $\bO(t + \lambda)$. Therefore,
    \begin{equation*}
        \dist_{s}(\rho_{\Ext(T, S_i) S_i E_{i+1}R_{i+1}}, \rho_{U_m} \otimes \rho_{S_i E_{i+1}R_{i+1}}) \le \varepsilon_{\mathrm{ext}} + 2\varepsilon' \;,
    \end{equation*}
    as required. %
\end{proof}

\begin{remark}
    Throughout the analysis, we are giving the adversary access to a hypothetical purifying system $R$. Note that this can only benefit the adversary, as losing this extra system cannot help in distinguishing the outputs from uniform randomness. 
\end{remark}

\section{Open Questions}\label{sec:open_questions}
Our work has introduced quantum unpredictability entropy and demonstrated its applications to cryptographic constructions. While we have established several important properties and applications, many interesting questions remain open. Below, we discuss several directions for future research. We organize them by subjects in the following subsections.

\subsection{Computational Purified Distance}
Is there a natural computational version of the purified distance?

The purified distance enjoys several desirable properties, most notably the data-processing inequality and what we refer to as the Uhlmann extension property. However, it is an information-theoretic measure: it does not reflect the computational limitations of the distinguisher. In contrast, in the classical setting of unpredictability entropy~\cite{HLR07SepPseudoentropyfromCompressibility}, both the indistinguishability and the guessing probability are defined relative to computationally bounded adversaries.

In the classical setting, using computational distance for smoothing entropy has benefits. It makes it trivial to see that for any $(s,\varepsilon)$ and any $XY$ classical joint distribution  $H_{\textrm{unp}}^{s,\varepsilon}(X|Y)\ge H^{\varepsilon,s}_{\hill}(X|Y)$. As a consequence, the classical unpredictability entropy of the output of pseudo-random generators on random short seeds, is high. In contrast, our definition of unpredictability would not assign high entropy to the image of pseudo-random generators, as the image is \emph{not} close to the uniform distribution in purified distance, only in computational distance.

A computational distance measure that retains the benefits of the purified distance for quantum states could lead to improved definitions of quantum computational entropies, particularly those involving smoothing or duality, and pseudo-random generators.

Recall that~\cref{lem:uniq_fidlity_for_pure_uhlmann}, adapted from~\cite{tan2024prospectsdeviceindependentquantumkey}, shows that fidelity is essentially the unique function satisfying both data-processing and the pure-Uhlmann property. As purified distance is a monotonic function of fidelity, this highlights its unique status as a metric compatible with quantum smoothing.

Thus, defining a computationally meaningful variant of the purified distance that preserves these properties seems inherently delicate. One possibility, as mentioned in the introduction, is to define a computational version $\bP^O$ via oracle access to a circuit that implements Uhlmann transformations. However, care must be taken: if the oracle is too powerful, the definition risks collapsing back to the information-theoretic case. This challenge is closely related to the complexity of the Uhlmann transformation problem~\cite{bostanci2023unitary}, and remains a fascinating open direction.

\subsection{Quantum Computational Entropies}

In this work, we defined unpredictability entropy for cq-states $\rho_{XE}$, where $X$ is classical and $E$ may be quantum, based on the guessing probability of a computationally bounded adversary. While this setting is natural for cryptographic tasks such as randomness extraction, and security of classical cryptographic protocols against quantum adversaries, it is not the most general. In a follow-up work~\cite{avidan2025fully}, we extend the definition of computational unpredictability entropy to fully quantum states, where both systems may be quantum, and we develop appropriate operational interpretations, including a generalization of the dual entropy.

In the information-theoretic setting, several seemingly distinct tasks, such as guessing a classical variable, decoupling a quantum system, or quantum correlations, are all quantified by the same entropy measure: smooth min-entropy~\cite{KRS09OperationalMeaningEntropy}. In the computational setting, however, it remains unclear whether the corresponding computational entropy notions (e.g., unpredictability entropy, correlation entropy, and decoupling-based entropy) coincide, or whether they define fundamentally different quantities. Understanding these relationships could clarify the landscape of quantum computational entropy and help identify the right tools for various cryptographic tasks.

Another compelling direction is to explore the connection between quantum computational entropies and quantum pseudorandom objects. Unpredictability is closely related to pseudo-randomness and the computational complexity of guessing, and it is natural to ask how measures like $\hunp$ relate to recent constructions of quantum pseudorandom states~\cite{QPRSJKF18}, unpredictable state generators~\cite{morimae2024quantumunpredictability}, quantum one-way puzzles~\cite{khurana2024commitmentsquantumonewayness}, pseudorandom unitaries~\cite{FermiMaHsin24pseRandom_Unitaries}, and pseudo-entangled states~\cite{arnonfriedman2023computationalentanglementtheory,aaronson2023quantumpseudoentanglement,bouland2023publickeypseudoentanglementhardnesslearning}.

\subsection{Pseudo-Randomness Extraction}
Beyond the inner-product function, are there other single-bit extractors that satisfy the reconstruction property required for Trevisan's extractor and remain secure against quantum adversaries for sources with high~$\hunps^\varepsilon(X|E)$?

Are there short-seed or two-source extractors that can extract pseudo-randomness from sources with high quantum computational unpredictability entropy?

More generally, what structural properties must extractors satisfy in order to work with sources quantified by $\hunps^\varepsilon(X|E)$ in the presence of quantum side-information? For instance, do such extractors require a special quantum reconstruction guarantee, or can classical reconstruction properties be adapted to ensure quantum security?

In the classical setting, it is known that all reconstructive extractors with an efficient reconstruction, such as Trevisan's extractor, work with unpredictability entropy~\cite{HLR07SepPseudoentropyfromCompressibility}. However, in the quantum case, no general class of extractors is currently known to be secure against quantum side-information in the unpredictability setting. Closing this gap remains an important direction for future work.

Even in the information-theoretic setting, it is not trivial that all classical reconstructive extractors are also quantum proof extractors with similar parameters, as is the case for Trevisan's extractors.
Not even ones with an efficient reconstruction process, such as in~\cite{Ext_recon_UR05,Ext_recon_UMANS2003419}. A reconstruction process may use the classical side-information multiple times.
Quantum side-information is more delicate; measuring it in one basis may destroy the information that could have been available on a different basis. For Trevisan's extractors, this problem was circumvented~\cite{DPVR12trevisanextwithquantumsideinfo} by a reduction to $1$ bit extractors that allows for a single measurement reconstruction. In our case, we show that a similar reduction is possible in the computational setting, and we show a single measurement \emph{efficient} reconstruction for the inner product $1$ bit extractor. More general questions about quantum reconstructive extractors remain open. Can any classical reconstruction process be transformed into a quantum restriction process? Is there a special property required for such reductions to such single measurement quantum reconstruction?

\section{Summary}

We introduce a new quantum computational unpredictability entropy measure that quantifies how difficult it is for computationally bounded quantum adversaries to predict classical secrets given quantum side-information. This framework allows us to bridge concepts from classical computational entropy with the structure and constraints of quantum information. Our work yields several contributions to quantum cryptography.

First, in~\cref{subsec:chain_rule_for_unp}, we prove a quantum leakage chain rule for this entropy measure. For any~$\rho_{XEL}$ cqq-state, $\varepsilon \ge 0$, $s \in \mathbb{N}$, and $\ell = \log \dim(L)$, we have:
\begin{equation*}
    \hunps^{\varepsilon}(X|EL)_{\rho} \ge \hunps[s + \bO(\ell)]^{\varepsilon}(X|E)_{\rho} - 2\ell \;.
\end{equation*}

Second, in~\cref{sec:ps_rand}, we show that certain randomness extractors can securely extract pseudo-randomness from sources with high unpredictability entropy, even in the presence of quantum side-information. We prove that the inner-product function serves as an effective one-bit extractor and show how to extend this to multi-bit extraction using Trevisan's construction~\cite{trevisan2001extractors} with only a modest increase in seed length.

Third, in~\cref{subsec:OCL}, we propose a new, more general model of quantum leakage channels that captures a wider class of quantum attacks than previous models. Our framework accounts for preexisting quantum side-information and removes the bounded-storage assumption typically imposed in prior work. We demonstrate the utility of the leakage chain rule and the leakage model by analyzing a protocol of alternating extraction under quantum leakage, both in the information-theoretic setting and the computational unpredictability setting.

A natural next step is to extend these ideas to settings where the computation itself is quantum. In such scenarios, one could potentially define and analyze the security of fully quantum protocols with quantum side-information, using the fully quantum version of unpredictability entropy introduced in our follow-up work~\cite{avidan2025fully}.

\paragraph{Acknowledgments.}
We thank Shafi Goldwasser for introducing us to the topic of computational entropies and thank Zvika Brakerski, Thomas Hahn, Amnon Ta-Shma and Thomas Vidick for useful discussions and Christian Schaffner and anonymous referees for their useful comments on a previous version of the manuscript.
This research was supported by the Israel Science Foundation (ISF), and the Directorate for Defense Research and Development (DDR\&D), grant No. 3426/21, the Peter and Patricia Gruber Award and by the Air Force Office of Scientific Research under award number FA9550-22-1-0391.
RA was further generously supported by the Koshland Research Fund and is the Daniel E. Koshland Career Development Chair.

\newpage

\appendix


\subsection{Chain Rule for Classical Unpredictability Entropy}\label{sec:appendix_chainrule}

Recall the fully classical chain rule for unpredictability entropy, the proof is similar to~\cite[Lemma 11]{KPWW14CounterexampleHILLchainrule}, rewritten in our notations and definitions.

\begin{lemma}\label{lem:classical-chain-rule-for-unpredictability}
    For $X,Y,L$ random variables where $L$ is distributed over $\bits^{\ell}$. Let $s \in \mathbb{N},\varepsilon\ge 0$ it holds that
    \begin{equation*}
        \hunp^{s + \bO(\ell) ,\varepsilon}(X|Y) \ge \hunp^{s,\varepsilon}(X|YL) - \ell
    \end{equation*}
\end{lemma}

\begin{proof}
    Let $k\in \mathbb{N}$ and let $t$ be the size of a circuit required to generate a uniformly random bit string of length $\ell$. Note that $t = \bO(\ell)$.
    We will show that
    \begin{equation}\label{eq:appendix-fully-classical-chain-rule-assumption}
        \hunp^{s,\varepsilon}(X|YL) < k - \ell \implies \hunp^{s + t,\varepsilon}(X|Y) < k \;
    \end{equation}
    The assumption in the LHS of~\cref{eq:appendix-fully-classical-chain-rule-assumption} implies that for any random variables $(W,Z,M)$ such that $\dist_{s}(XYL,WZM) <\varepsilon$ there is a circuit $C$ of size $s$ such that
    \begin{equation*}
        \pr{C(ZM) = W} > 2^{-k + \ell} \;.
    \end{equation*}
    Given $(W,Z)$ such that $\dist_{s + t}(XY,WZ) < \varepsilon$, we know that $\dist_{s}(XY,WZ) \le \dist_{s + t}(XY,WZ) < \varepsilon$, since bigger circuits can only help in distinguishing and that there is an extension of $(W,Z)$ to a joint probability $(W,Z,M)$ such that $\dist_{s}(XY,WZ) = \dist_{s}(XYL,WZM)$.
    We can define a circuit $C'$ that takes $z\in Z$ as input, generates uniformly at random $\ell$ bits, $l$, and then outputs $C(z,l)$.
    The size of $C'$ is $s + t$. For $(w,m,z)$ chosen from $(W,Z,M)$ and $l$ chosen uniformly at random we have
    \begin{align*}
        \pr{C'(z) = w} & \ge \pr{C'(z) = w | m = l} \cdot \pr{m = l} \\
                       & \ge \pr{C(z,m) = w} \cdot 2^{-\ell}         \\
                       & \ge 2^{-k + \ell} \cdot 2^{-\ell}           \\
                       & = 2^{-k} \;.
    \end{align*}
    This implies that $\hunp^{s + t,\varepsilon}(X|Y) < k$ as required. %
\end{proof}

We can extend this definition to a more general definition of unpredictability entropy by separating the computational parameter into two computational parameters representing the two roles $s$ has in the original definition.
$s_{ind}$ for $(\varepsilon, s_{ind})$ indistinguishably and $s_{guess}$ for the size of the guessing circuit.

\begin{definition}\label{def:two-parameter-unpredictability-entropy}
    For any cq-state $\rho_{XE}$, and $\varepsilon\ge 0$, $s_{ind},s_{guess}\in\bbN$.
    We say that
    \begin{equation*}
        \hunp^{s_{ind},s_{guess},\varepsilon}(X|E)_{\rho}\ge k \;,
    \end{equation*}
    if there is a cq-state $\tilde{\rho}_{XE}$ such that $\dist_{s_{ind}}(\rho_{XE},\tilde{\rho}_{XE}) \le \varepsilon$, for any guessing circuit $\cC$ of size $s_{guess}$
    \begin{equation*}
        \pr{\cC(\tilde{\rho}_{E}^x) = x} \le 2^{-k} \;.
    \end{equation*}
\end{definition}

\begin{lemma}\label{lem:classical-chain-rule-two-parameter}
    For $X,Y,L$ random variables where $L$ is distributed over $\bits^{\ell}$. Let $s_{ind},s_{guess} \in \mathbb{N},\varepsilon\ge 0$ it holds that
    \begin{equation*}
        \hunp^{s_{ind},s_{guess} + \bO(\ell),\varepsilon}(X|Y) \ge \hunp^{s_{ind},s_{guess},\varepsilon}(X|YL) - \ell
    \end{equation*}
\end{lemma}

\begin{proof}
    Let $k\in \mathbb{N}$ and let $t$ be the size of a circuit required to generate a uniformly random bit string of length $\ell$.
    Assume that
    \begin{equation*}
        \hunp^{s_{ind},s_{guess},\varepsilon}(X|YL) < k - \ell \;.
    \end{equation*}
    By definition, for any $(W,Z,M)$ such that~$\dist_{s_{ind}}(XYL,WZM)<\varepsilon$ there is a circuit $C$ of size $s_{guess}$ such that
    \begin{equation*}
        \pr{C(ZM) = W} > 2^{-k + \ell} \;.
    \end{equation*}
    Given $(W,Z)$ such that $\dist_{s_{ind}}(XY,WZ)<\varepsilon$, we can define a circuit $C'$ that:
    \begin{enumerate}
        \item Takes $z\in Z$ as input.
        \item Generates uniform $l \in \{0,1\}^\ell$.
        \item Outputs $C(z,l)$.
    \end{enumerate}
    The size of $C'$ is $s_{guess} + t$. The success probability of $C'$ for $(w,m,z)$ chosen from $(W,Z,M)$ and $l$ chosen uniformly at random we have that
    \begin{align*}
        \pr{C'(z) = w} & \ge \pr{C(z,m) = w} \cdot 2^{-\ell}            \\
                       & \ge 2^{-k + \ell} \cdot 2^{-\ell} = 2^{-k} \;.
    \end{align*}
    This is true since there is always a classical extension that does not increase the computational distance and
    \begin{equation*}
        \dist_{s_{ind}}(XY,WZ) = \dist_{s_{ind}}(XYL,WZL) < \varepsilon\;.
    \end{equation*}
    This implies $\hunp^{s_{ind},s_{guess} + t,\varepsilon}(X|Y) < k$. %
\end{proof}

\subsection{Inner-Product Extractor}\label{sec:appendix_IP_ext}

Proof of~\cref{lem:fromIPtoX}, the proof is similar to~\cite[Lemma 14]{KK10twosourceextractorssecurequantum}. First, we will show the exact version. Assuming that an adversary can guess the inner-product exactly for every $y$ using the same quantum side-information, we will show that it can reconstruct $x$ exactly using the same quantum side-information.

\begin{lemma}\label{lem:clean-comp}
    If a circuit $\cC$ of size $s$ can guess $\IP(x,y)$ using $\rho_{E}^{x}$ with probability $1$, then there is a circuit $\cC'$ of size $2s+1$ that implements the following unitary:
    \begin{equation*}
        U( \ket{z} \ket{\rho_{E}^{x}} \ket{y}\ket{0}) = \ket{z \oplus \IP(x,y)}\ket{\rho_{E}^x}\ket{y}\ket{0} \;.
    \end{equation*}
\end{lemma}

The construction of $\cC'$ is simple; instead of measuring the qubit that holds the inner-product after applying $\cC$, we apply $\CNOT$ controlled by this qubit to the new qubit holding $\ket{z}$ followed by running the inverse $\cC$.

The construction of $\cG$ from that is also simple, but the proof that it works is a bit more complicated. To recover all of $x$, we create a superposition on all the possible $y$'s and 1 ancilla qubit (taking $n+1$ gates)
\begin{equation*}
    \sum_{a\in \bits, y\in\bits^n} (-1)^a \ket{a} \ket{y} \ket{\rho_{E}^x} \ket{0} \;.
\end{equation*}
We apply $\cC'$ on this state and then apply Hadamard gates on the first $n+1$ qubits.

Here we construct the reduction proving~\cref{lem:fromIPtoX}. Suppose there exists a distinguishing circuit $\mathcal{C}$ that, given access to the inner product, can distinguish its output from a uniformly random bit with non-negligible advantage. We will explicitly build a guessing circuit $\mathcal{G}$ that (i) successfully guesses the input and (ii) uses only slightly more gates than $\mathcal{C}$. For convenience, we first recall the formal statement of~\cref{lem:fromIPtoX}:

\begin{lemma}\label{appendix:lem:fromIPtoX}
    Let $\rho_{XE}$ be a cq-state.
    If there is a circuit $\cC$ of size $s$ that can guess $\IP(x,y)$ using $\rho_{E}^{x}$ with probability $\frac{1}{2} + \varepsilon$, where the probability is over the distribution of $x$ and a uniformly random $y$. Then there is a circuit $\cG$ of size $2(s+n+1)$ that can guess $x$ using $\rho_{E}^{x}$ with probability $4\varepsilon^{2}$.
\end{lemma}

The proof is similar to~\cite[Theorem 12]{KK10twosourceextractorssecurequantum}. The main differences are in notations and the fact that we care about the complexity of a single circuit, and they care about the communication complexity of two parties using local circuits and entanglement.

For clarity and completeness, we will write the full proof with our notations.

The proof can be split into three parts. First, we show that if there is a circuit that guesses the inner-product of $x$ exactly with every $y$, then there is a circuit that can perfectly guess all of $x$.
Second, we will show what happens when running this circuit when the probability of guessing the inner-product is less than 1. Finally, we will show that the inner-product is a good extractor against quantum side-information with conditional unpredictability entropy.

\begin{lemma}
    For any channel $\cC$ and any function such that
    \begin{equation*}
        \cC(\ket{0} \ket{x,y}\ket{0}) = \ket{f(x,y)}\ket{R} \;,
    \end{equation*}
    for some $R$, there is a channel $\cC'$ such that
    \begin{equation*}
        \cC'(\ket{z}\ket{x,y}\ket{0}) = \ket{z\oplus f(x,y)}\ket{x,y}\ket{0} \;,
    \end{equation*}
    where $\oplus$ denotes bitwise addition mod 2. The size of $\cC'$ is at most $2|\cC|+1$.
\end{lemma}

\begin{proof}
    In our computational model, any guessing circuit can be written as a unitary followed by a measurement in the computational basis. We also assumed that any gate in our universal gate set has an inverse in the set. The circuits $\cC'$ is thus simply a composition of the unitary part of $\cC$ on the last qubits (not including the new qubit initialized to $z$) followed by $\CNOT$ controlled by the qubit that is to be measured by $\cC$ on $\ket{z}$, and then the inverse of $\cC$. %
\end{proof}

\begin{lemma}\label{lem:inner_prod_circ_to_reconstcutXcirc}
    Let $\rho_{XE}$ be a cq-state such that $\rho_{X}$ is a distribution over $\bits^{n}$. Let $\cC$ be a circuit of size $s$ such that for every $y\in \bits^{n}$
    \begin{equation*}
        \cC(\ket{0}\ket{\rho^{x}_{E}}\ket{y}\ket{0}) = \ket{\IP(x,y)}\ket{K_{x,y}}\ket{y}\ket{0}\;.
    \end{equation*}
    For $\ket{\rho_{E}^x}$ some purification of $\rho_{E}^x$ and $\ket{K_{x,y}}$ some pure state.
    Then there is a circuit $\cR$ of size at most $2s+2n+5$ that reconstructs $x$ exactly from $\ket{\rho_{E}^x}$

    \begin{equation*}
        \cR(\ket{0}^{\otimes(n+1)}\ket{\rho^{x}_{E}}\ket{0}) = \ket{1}\ket{x}\ket{\rho^{x}_{E}}\ket{0} \;.
    \end{equation*}
\end{lemma}

\begin{proof}
    We will explicitly construct $\cR$ from $\cC$, the construction of $\cR$ from $\cC,\cC^{\dagger}$ and basic gates is illustrated in~\cref{fig:Circuit-IPtoX}.

    We write the state of the system after each layer of the circuit is applied, counting the number of gates and showing that it can guess $x$ along the way.
    \begin{enumerate}
        \item Preparing a state $\ket{1}\ket{0}^{\otimes n}\ket{\rho^{x}_{E}}$. ($0$ gates, or $1$ NOT gate if we can only create $\ket{0}$)
        \item Applying Hadamard to the first $n+1$ qubit, using $n+1$ $\mathrm{H}$ gates, resulting in a state
              \begin{equation*}
                  \frac{1}{\sqrt{2^{n+1}}} \sum_{a\in \bits , y \in \bits^{n} } (-1)^{a} \ket{a} \ket{y} \ket{\rho^{x}_{E}}\;.
              \end{equation*}
        \item Performing $\cC$ on this state, on all but the first qubit, using $s$ gates and CNOT to the first qubit, controlled by the result qubit of $\cC$, resulting in
              \begin{equation*}
                  \frac{1}{\sqrt{2^{n+1}}} \sum_{a\in \bits , y \in \bits^{n} } (-1)^{a} \ket{a \oplus \IP(x,y)} \ket{y} \ket{K_{x,y}}\;.
              \end{equation*}
        \item Performing the inverse circuit $\cC^{\dag}$ on this state, using $s$ gates resulting in the state
              \begin{equation*}
                  \frac{1}{\sqrt{2^{n+1}}} \sum_{a\in \bits , y \in \bits^{n} } (-1)^{a} \ket{a \oplus \IP(x,y)} \ket{y} \ket{\rho^{x}_{E}}\;.
              \end{equation*}
        \item Applying Hadamard to the first $n+1$ qubits, using $n+1$ $\mathrm{H}$ gates, resulting in a state
        \begin{align*}
        H^{\otimes(n+1)} & \frac{1}{\sqrt{2^{n+1}}} \sum_{a\in \bits , y \in \bits^{n} } (-1)^{a} \ket{a \oplus \IP(x,y)} \ket{y}   \\
          & = \frac{1}{2^{ n+1}}
             \sum_{a,y}\sum_{b\in\bits}\sum_{z\in\bits^{n}}
               (-1)^{a} 
               (-1)^{(a\oplus\IP(x,y))b}
               (-1)^{y\cdot z} 
               \ket{b}_{A}\ket{z}_{B}\\
          &= \frac{1}{2^{ n}}
             \sum_{y,z}
               (-1)^{\IP(x,y)}
               (-1)^{y\cdot z} 
               \ket{1}_{A}\ket{z}_{B}\\
          &= \frac{1}{2^{ n}}
             \sum_{z}
               \bigl(\sum_{y}(-1)^{y\cdot(x\oplus z)}\bigr)
               \ket{1}_{A}\ket{z}_{B}\\
          &= \ket{1}_{A}\ket{x}_{B}\;.
        \end{align*}
        The qubits in register $E$ remain unchanged in this step, therefore the resulting state is
              \begin{equation*}
                  \ket{1}\ket{x}\ket{\rho^{x}_{E}}\;.
              \end{equation*}
        \item Measuring the first $n+1$ qubits in the computational basis, resulting in the measurement result $1,x$. ($0$ gates since measuring in the computational basis is not counted in our complexity measure.) %
    \end{enumerate}
\end{proof}

\begin{lemma}\label{lem:inner_prod_aprox_toX_apendix}
    Let $\rho_{XE}$ be a cq-state such that $\rho_{X}$ is a distribution over $\bits^{n}$. Let $\cC$ be a circuit of size $s$ such that for every $y\in \bits^{n}$
    \begin{equation*}
        \cC(\ket{0}\ket{\rho^{x}_{E}}\ket{y}\ket{0}) = \alpha_{x,y}\ket{\IP(x,y)}\ket{G_{x,y}}\ket{y}\ket{0}+\beta_{x,y}\ket{\overline{\IP(x,y)}}\ket{B_{x,y}}\ket{y}\ket{0}\;.
    \end{equation*}
    For $\ket{\rho_{E}^x}$ some purification of $\rho_{E}^x$ and $\ket{G_{x,y}}, \ket{B_{x,y}}$ some pure states and
    $$\mathbb{E}_{y} \left[\beta^{2}_{x,y}\right] = \frac{1}{2}-\varepsilon_{x}\;.$$
    Then there is a circuit $\cR$ of size at most $2s+2n+5$ that reconstructs $x$ from $\ket{\rho_{E}^x}$ with probability at least $4\varepsilon_{x}^2$.
\end{lemma}

\begin{proof}
    The circuit is the same as the exact proof, the states of the system after each step are more complicated and require some more care.
    We write the state of the system after each layer of the circuit is applied, counting the number of gates and showing that it can guess $x$ along the way.
    \begin{enumerate}
        \item Preparing a state $\ket{1}\ket{0}^{\otimes n}\ket{\rho^{x}_{E}}$. ($0$ gates, or $1$ NOT gate if we can only create $\ket{0}$)
        \item Applying Hadamard to the first $n+1$ qubit, using $n+1$ $\mathrm{H}$ gates, resulting in a state
              \begin{equation*}
                  \frac{1}{\sqrt{2^{n+1}}} \sum_{a\in \bits ,  y \in \bits^{n} } (-1)^{a} \ket{a} \ket{y} \ket{\rho^{x}_{E}}\;.
              \end{equation*}
        \item Performing $\cC$ on this state, on all but the first qubit, using $s$ gates and CNOT to the first qubit, controlled by the result qubit of $\cC$, resulting in
              \begin{align*}
                  \frac{1}{\sqrt{2^{n+1}}} \sum_{a\in \bits , y \in \bits^{n} }
                  \Big( (-1)^{a} \ket{a \oplus \IP(x,y)} \alpha_{x,y} \ket{y} \ket{G_{x,y}} \\
                  + \beta_{x,y}\ket{a \oplus \overline{\IP(x,y)}}\ket{y}\ket{B_{x,y}}\Big) \;.
              \end{align*}
              Each element of the sum can be written as:
              \begin{multline}
                  \ket{a+\IP(x,y)}
                  \left( \alpha_{x,y} \ket{y}\ket{G_{x,y}}+\beta_{x,y}\ket{y}\ket{B_{x,y}}\right)
                  \\
                  +\sqrt{2}\beta_{x,y}
                  \left( \frac{1}{\sqrt{2}}\ket{a\oplus \overline{\IP(x,y)}}-\frac{1}{\sqrt{2}}\ket{a\oplus \IP(x,y)}\right) \ket{y}\ket{B_{x,y}}
              \end{multline}
        \item Performing the inverse circuit $\cC^{\dagger}$ on this state, using $s$ gates resulting in the state
              \begin{equation*}
                  \frac{1}{\sqrt{2^{n+1}}}  \sum_{a\in \bits , y \in \bits^{n} } (-1)^{a} \ket{a \oplus \IP(x,y)} \ket{y} \ket{\rho^{x}_{E}}
                  +(-1)^{a}\sqrt{2}\beta_{x,y}\ket{M_{x,y,a}}\;,
              \end{equation*}
              where $M_{x,y,a}=\left(\frac{1}{\sqrt{2}}\ket{a+\overline{IP(x,y)}}-\frac{1}{\sqrt{2}}\ket{a+IP(x,y)}\right) \cC^{\dagger}\ket{y}\ket{B_{x,y}}$
        \item Applying Hadamard to the first $n+1$ qubits, using $n+1$ $\mathrm{H}$ gates.
              We can avoid fully writing the state after applying Hadamard, since Hadamard does not change the size of the error term, we can write the state of the system at this step in the exact protocol as:
              \begin{equation*}
                  \ket{g_{x}} = \frac{1}{\sqrt{2^{n+1}}}  \sum_{a\in \bits , y \in \bits^{n} } (-1)^{a} \ket{a \oplus \IP(x,y)} \ket{y} \ket{\rho^{x}_{E}}\;.
              \end{equation*}
              The error from the exact protocol is
              \begin{align*}
                  \ket{e_{x}} & = \frac{1}{\sqrt{2^{n+1}}}  \sum_{a\in \bits , y \in \bits^{n} } (-1)^a \sqrt{2}\beta_{x,y}\ket{M_{x,y,a}} \\
                              & = \frac{1}{\sqrt{2^{n}}}  \sum_{y \in \bits^{n}} \sqrt{2}\beta_{x,y}\ket{M_{x,y,0}}\;.
              \end{align*}
              We can see that
              \begin{equation*}
                  \bra{g_{x}}(\ket{g_{x}}+\ket{e_{x}}) = 2\varepsilon_{x}\;.
              \end{equation*}
        \item Measuring the first $n+1$ qubits in the computational basis, resulting in the measurement result $1,x$ with probability $ 4\varepsilon_{x}^{2}$. ($0$ gates since measuring in the computational basis is not counted in our complexity measure.) %
    \end{enumerate}
\end{proof}

\begin{figure}
    \centering
    \input{Circuit-IPtoX}
    \caption{Structure of a circuit that Guesses $X$ using $\cC$ a circuit that predicts the inner-product $\IP(x,y)$, using $y$ and quantum side information about $x$, $\sigma_{x}$.}
    \label{fig:Circuit-IPtoX}
\end{figure}

From this, we can recover a known lemma about the inner-product in the information-theoretic case.

\begin{lemma}[Corollary 14~\cite{KK10twosourceextractorssecurequantum}]
    For any $\varepsilon_{\ext}>0$ and $k_{\ext}> 1-2\log(\varepsilon)$ $\IP(x,y)$ is a $(k_{\ext},\varepsilon_{\ext})$ extractor against quantum side-information with uniform seed.
\end{lemma}

But more importantly for us, we can get a computational version with unpredictability entropy.

\begin{lemma}[Inner-Product Extractor from Unpredictability]\label{apendix:lem_IPgoodQExt}
    Let $\rho_{XE}$ be a cq-state where $\rho_{X}$ is a distribution over $\bits^{n}$. Let $\rho_{Y}$ be maximally mixed over $n$ qubits. Let $k_{\ext} \in \mathbb{N}, \varepsilon_{\ext}>0$ such that $k_{\ext} \ge 1-2\log(\varepsilon_{\ext})$. If
    \begin{equation*}
        \hunps[2s+2n+5]^{\varepsilon}(X|E) \ge k_{\ext} \;,
    \end{equation*}
    then
    \begin{equation*}
        \dist_{s}(\rho_{\IP(X,Y)YE},\rho_{U_{1}}\otimes\rho_{YE}) \le \varepsilon_{\ext} + 2\varepsilon \;.
    \end{equation*}
\end{lemma}

\begin{proof}
    Let $\tilde{\rho}_{XE}$ be a cq-state such that $\bP(\tilde{\rho}_{XE},\rho_{XE})\le \varepsilon$. Assume for contradiction that
    \begin{equation*}
        \dist_{s}(\rho_{\IP(X,Y)YE},\rho_{U_{1}}\otimes\rho_{YE}) > \varepsilon_{\ext} + 2\varepsilon \;.
    \end{equation*}
    Using the triangle inequality for computational distance~\cref{lemma:triangle-inequality-computational-td} we get that
    \begin{align*}
        \dist_{s}(\rho_{\IP(X,Y)YE},\tilde{\rho}_{\IP(X,Y)YE})
         & + \dist_{s}(\tilde{\rho}_{\IP(X,Y)YE} ,\rho_{U_{1}}\otimes \tilde{\rho}_{YE})                                          \\
         & + \dist_{s}(\rho_{U_{1}}\otimes\tilde{\rho}_{YE},\rho_{U_{1}}\otimes\rho_{YE}) > \varepsilon_{\ext} + 2\varepsilon \;.
    \end{align*}
    From the triangle inequality,
    \begin{equation*}
        \dist_{s}(\tilde{\rho}_{\IP(X,Y)YE} ,\rho_{U_{1}}\otimes\rho_{YE}) > \varepsilon_{\ext} \;.
    \end{equation*}
    From~\cref{lemma:one_bit_distinguish_is_predict} we know that this implies there is a guessing circuit that can guess $\tilde{\rho}_{\IP(X,Y)}$ from $\rho_{YE}$ with probability at least $\frac{1}{2}+\varepsilon_{\ext}$. From~\cref{lem:inner_prod_aprox_toX_apendix} that means there is a circuit of size at most $s+2n+5$ that guess $\tilde{\rho}_{X}^{x}$ from $\rho^{x}_{E}$ with probability at least $4\varepsilon_{\ext}^{2}$ in contradiction to the assumption
    \begin{equation*}
        \hunps[2s+2n+5]^{\varepsilon}(X|E) \ge 1-2\log(\varepsilon_{\ext}) \;. \qedhere
    \end{equation*}
\end{proof}

\subsection{Computational Proof of Trevisan Extractors}\label{sec:appendix_trevisan}

In this section, we provide a complete proof of~\cref{lem:m-bit-ext-from-any-1-bit-and-design}. Recall

\begin{theorem}
    Let $C':\bits^{n}\times\bits^{t} \to \bits$ be a $(k,\varepsilon)$-one-bit extractor secure against $s$-unpredictability entropy.
    Let $S_{1},\dots,S_{m} \subset [d]$ be a weak $(t,r)$-design. Define the following function:
    \begin{align}
        \Ext_{C}: \bits^{n} \times \bits^{d} & \to \bits^{m}                                                                 \\
        (x,y)                                & \mapsto \left(C(x,y_{S_{1}}) \;,\dots \;, C(x,y_{S_{m}}) \right)\;, \nonumber
    \end{align}
    where $y_{S}$ is the bits of $y$ in locations $S$.
    $\Ext_{C}$ is a $\left(k + rm - \log(\varepsilon), 2 m \varepsilon\right)$ extractor of pseudorandom bits for quantum unpredictability entropy in the following sense: If
    \begin{equation*}
        \hunps[s']^{\varepsilon'}(X|E)_{\rho} \ge k + rm - \log(\varepsilon) \;,
    \end{equation*}
    then
    \begin{equation*}
        \dist_{s}(\rho_{\Ext_{C}(X,Y)YE}, \rho_{U_{m}} \otimes \rho_{Y} \otimes \rho_{E}) \le 2 m \varepsilon + 2\varepsilon' \;,
    \end{equation*}
    where $s' = \bO(ns+rm)$.
\end{theorem}

First, we will reduce the problem from many bits to one bit, by showing that if a distribution on bitstrings is computationally far from the uniform distribution on bitstrings, there is at least one bit that is far from a uniform bit, given the previous bits.
Using the triangle inequality, we can show the following lemma
\begin{lemma}\label{lem:from_many_bits_to_one}
    Let $\rho_{ZB}$ be a cq-state, where $Z$ is a classical random variable over $m$ bit strings. If
    \begin{equation*}
        \dist_{s}(\rho_{ZB},\rho_{U_{m}} \otimes \rho_{B}) > \varepsilon \;,
    \end{equation*}
    then there is a bit $i\in [m] $ such that
    \begin{equation*}
        \dist_{s} \left(
        \sum_{z\in Z} p_{z} \proj{z_{[i-1]}0} \otimes \rho_{B}^{z}  , \quad
        \sum_{z\in Z} p_{z} \proj{z_{[i-1]}1} \otimes \rho_{B}^{z}   \right) > \frac{\varepsilon}{m} \;.
    \end{equation*}
    The notation $z_{[i-1]}$ denotes the first $i-1$ bits of the string $z$.
\end{lemma}

We can interpret this lemma as saying that if we can distinguish the state from uniform randomness with a circuit of size $s$, there is a bit that we can guess with advantage at least $\varepsilon/m$ with a circuit of size $s$ using the same classical information and the bits before it in the string.

\begin{proof}
    The proof is almost identical to the proof in~\cite{DPVR12trevisanextwithquantumsideinfo}, since the hybrid argument they use is based on the triangle inequality for trace distance. We only need to replace trace distance by computational distance and repeat the proof.
    Let:
    \begin{equation*}
        \sigma_{i} = \sum_{\substack{z \in Z \\ r \in \bits^m}} \frac{p_{z}}{2^{m}} \proj{z_{[i]},r_{\set{i+1,\dots,m}}} \otimes \rho_{B}^{z} \;.
    \end{equation*}
    Then:
    \begin{align*}
        \varepsilon & < \dist_{s} (\rho_{ZB},\rho_{U_{m}} \otimes \rho_{B}) \\
                    & = \dist_{s} (\sigma_{m},\sigma_{0})                   \\
                    & \le \sum_{i=1}^{m} \dist_{s}(\sigma_{i},\sigma_{i-1}) \\
                    & \le m \max_{i} \dist_{s}(\sigma_{i},\sigma_{i-1}) \;.
    \end{align*}
    And by construction for every $i$:
    \begin{equation*}
        \dist_{s} (\sigma_{i},\sigma_{i-1}) = \dist_{s} (\rho_{Z_{[i]}B},\rho_{U_{i}Z_{[i-1]}B}) \;,
    \end{equation*}
    where $\rho_{U_{i}}$ is a maximally mixed on the $i$-th bit. And $\rho_{Z_{[i]}}$ is the reduced state of $\rho_{Z}$ with only the first $i$ bits of $Z$.
    By applying~\cref{lemma:one_bit_distinguish_is_predict} we get the desired result. %
\end{proof}

Following the proof from~\cite{DPVR12trevisanextwithquantumsideinfo}, we use the structure of $\Ext_{C}$ as a concatenation of $C$ using different parts of the seed in a weak design to bound the amount of information that the previous $i-1$ bits can contribute. We now show there is a way to split the seed $Y$ to $V,W$ and construct some classical advice system $G$ of size at most $rm$ such that they form a Markov chain $V \mch W \mch G$ and if
\begin{equation*}
    \Vert \rho_{\Ext_{C}(X,Y)E} - \rho_{U_{m}}\otimes\rho_{Y}\otimes\rho_{E} \Vert >\varepsilon\;,
\end{equation*}
then
\begin{equation*}
    \Vert\rho_{C(X,V)VWGE} - \rho_{U_1} \otimes \rho_{VWGE} \Vert > \varepsilon/m \;,
\end{equation*}
where $r$ is a parameter of the weak design and $m$ is the output size of the extractor.

\begin{lemma}\label{lem:reducing_big_ext_to_small_ext_with_advice}
    Let $\rho_{XE}$ be a cq-state, let $\rho_{Y}$ be a classical seed (not necessarily uniform) independent of $\rho_{XE}$.
    If
    \begin{equation}\label{eq:distinguish_big_ext}
        \dist_{s}(\rho_{\Ext_{C}(X,Y)E} , \rho_{U_{m}} \otimes \rho_{Y} \otimes \rho_{E}) > \varepsilon \;,
    \end{equation}
    then there is a fixed partition of $Y$ into $V,W$ and a classical advice system $G$ of size at most $rm$ such that $V \mch W \mch G$ and
    \begin{equation*}
        \dist_{s}(\rho_{C(X,V)VWGE} , \rho_{U_{1}} \otimes \rho_{VWGE}) > \frac{\varepsilon}{m} \;.
    \end{equation*}
\end{lemma}

\begin{proof}
    The proof is similar to the non-computational version~\cite[Proposition 4.4]{DPVR12trevisanextwithquantumsideinfo}. We use~\cref{lem:from_many_bits_to_one} on~\cref{eq:distinguish_big_ext} and the structure of $\Ext_{C}$ to get that there is a bit $i\in [m]$ such that:
    \begin{multline}
        \dist_{s}
        \left(
        \sum_{\substack{x,y \\ C(x,y_{S_{i}}) = 0 }} p_{x}q_{y} \proj{C(x,y_{S_{1}})\dots C(x,y_{S_{i-1}}),y} \otimes \rho^{x} , \right. \\
        \left.
        \sum_{\substack{x,y \\ C(x,y_{S_{i}}) = 1 }} p_{x}q_{y} \proj{C(x,y_{S_{1}})\dots C(x,y_{S_{i-1}}),y} \otimes \rho^{x}
        \right) > \frac{\varepsilon}{m} \;,
    \end{multline}
    where $\set{p_{x}},\set{q_{y}}$ are the classical probability distributions of $X,Y$. We can split any $y \in Y$ to strings of length $t$ and $d-t$ respectively, denote them $v= y_{S_{i}}$ , $w = y_{[d] \setminus S_{i}}$. Fixing $w,x,j$ and looking at $g(w,x,j,v) : = C(x,y_{s_{j}})$ as a function of $v$ ($g(w,x,j,\cdot) : \bits^{t}\to \bits$), this is a function that depends on $\abs{S_{j}\cap S_{i}}$ bits and has a single bit output, therefore it requires $2^{\abs{S_{j} \cap S_{i}}}$ bits.
    To describe $g^{w,x}(\cdot) := g(w,x,1,\cdot),\dots,g(w,x,i-1,\cdot)$ the concatenation of all the $i-1$ first bits, we need a string of length $\sum_{j=1}^{i-1} 2^{\abs{S_{j} \cap S_{i}}}$, which is at most $rm$ by the property of weak $(r,t)$-design.
    We denote the system that contains the string described by the functions $G$. Note that, given $W$, the advice system $G$ is independent of the random variable $V$, since it contains all the options for $v\in V$. Adding more information can only increase the computational distance, since a guessing circuit can simply not read the part that is not used for any given $v$, therefore $V\mch W \mch G$.
    We can rewrite the inequality as:
    \begin{equation*}
        \dist_{s}
        \left(
        \sum_{\substack{x,v,w \\ C(x,v) = 0 }} p_{x}q_{y} \proj{g^{x,w}(v),v,w} \otimes \rho^{x} ,
        \sum_{\substack{x,v,w \\ C(x,v) = 1 }} p_{x}q_{y} \proj{g^{x,w}(v),v,w} \otimes \rho^{x}
        \right) > \frac{\varepsilon}{m} \;.
    \end{equation*}
    Since providing all of $g$ can only increase the computational distance:
    \begin{equation*}
        \dist_{s}
        \left(
        \sum_{\substack{x,v,w \\ C(x,v) = 0 }} p_{x}q_{y} \proj{g^{x,w},v,w} \otimes \rho^{x} ,
        \sum_{\substack{x,v,w \\ C(x,v) = 1 }} p_{x}q_{y} \proj{g^{x,w},v,w} \otimes \rho^{x}
        \right) > \frac{\varepsilon}{m} \;.\qedhere
    \end{equation*}
\end{proof}

We are now ready to use the above lemmas to show we can extract pseudorandom bits from sources with high unpredictability entropy. A computational version of Theorem 4.6~\cite{DPVR12trevisanextwithquantumsideinfo}
\begin{theorem}
    Let $C':\bits^{n}\times\bits^{t} \to \bits$ be a $(k,\varepsilon)$-one-bit extractor secure against $s$-unpredictability entropy.
    Let $S_{1},\dots,S_{m} \subset [d]$ be a weak $(t,r)$-design. Defining the following function:
    \begin{align*}
        \Ext_{C}: \bits^{n} \times \bits^{d} & \to \bits^{m}                                            \\
        (x,y)                                & \mapsto C(x,y_{S_{1}})\dots C(x,y_{S_{m}}) \;, \nonumber
    \end{align*}
    where $y_{S}$ is the bits of $y$ in locations $S$.
    $\Ext_{C}$ is an extractor of pseudorandom bits for quantum unpredictability entropy in the following sense:
    If
    \begin{equation*}
        \hunps[s']^{\varepsilon'}(X|E)_{\rho} \ge k + rm - \log(\varepsilon) \;,
    \end{equation*}
    then
    \begin{equation*}
        \dist_{s}(\rho_{\Ext_{C}(X,Y)YE}, \rho_{U_{m}} \otimes \rho_{Y} \otimes \rho_{E}) \le 2 m (\varepsilon + \varepsilon') \;,
    \end{equation*}
    where $s' = \bO(ns+rm)$.
\end{theorem}

\begin{proof}[of~\cref{lem:m-bit-ext-from-any-1-bit-and-design}]
    Assuming that:
    \begin{equation*}
        \dist_{s}(\rho_{\Ext_{C}(X,Y)YE},\rho_{U_{m}} \otimes \rho_{Y} \otimes \rho_{E}) >  2 m (\varepsilon)  \;.
    \end{equation*}
    From~\cref{lem:reducing_big_ext_to_small_ext_with_advice} we know there is a way to split the seed $Y=VW$ and a classical advice $G$ of size at most~$rm$ such that:
    \begin{equation*}
        \dist_{s}(\rho_{C(X,V)WE},\rho_{U_{1}} \otimes \rho_{VGWE} ) > 2  (\varepsilon) \;.
    \end{equation*}
    From~\cref{lem:IPgoodQExt} and the chain rule for unpredictability entropy~\cref{lem:chain-rule-quantum-unpredictability-entropy} we get that:
    \begin{equation*}
        \hunps[\bO(ns+\log|G|)]^{\varepsilon'}(X|WE) \le \hunps[\bO(ns)]^{\varepsilon'}(X|WGE) + \log|G| < k+rm-\log(\varepsilon) \;.
    \end{equation*}
    If the seed $Y$ is uniformly random, then $W$ is also uniformly random and independent randomness does not improve guessing probability, and $\log(|G|) \le rm$ so:
    \begin{equation*}
        \hunps[\bO(ns+\log|G|)]^{\varepsilon'}(X|WE) \le \hunps[\bO(ns+rk)]^{\varepsilon'}(X|E) \;.
    \end{equation*}
    Therefore for uniform seed $Y$, if $\hunps[s']^{\varepsilon'}(X|E) \ge k+rm - \log(\varepsilon)$ then the output of $\Ext_{C}$ is pseudorandom with the given parameters
    \begin{equation*}
        \dist_{s}(\rho_{\Ext_{C}(X,Y)YE},\rho_{U_{m}} \otimes \rho_{Y} \otimes \rho_{E}) \le 2 m (\varepsilon+\varepsilon')  \;. \qedhere
    \end{equation*}
\end{proof}

\subsection{Relation to Previously Suggested Computational Entropies}
We also mention another way to define quantum computation entropy based on guessing the probability of bounded adversaries.

\begin{definition}[Quantum Conditional HILL Entropy~\cite{CC17Computationalminentropy}]\label{def:hill_cq-entropy}
    Let $\rho_{XE}$ be a cq-state, $s\in \mathbb{N}, \varepsilon\ge 0$. We say that
    \begin{equation*}
        H_{\hill}^{\varepsilon,s}(X|E)_{\rho} \ge k\;,
    \end{equation*}
    if there is cq-state $\tilde{\rho}_{XE}$ such that $\dist_{s}(\rho_{XE},\tilde{\rho}_{XE}) \le \varepsilon$ and
    \begin{equation*}
        H_{\min}(X|E)_{\tilde{\rho}} \ge k \;.
    \end{equation*}
\end{definition}

\begin{definition}[Quantum Guessing Pseudoentropy~\cite{CC17Computationalminentropy}]\label{def:guessing-pseudoentropy}
    Let $\rho_{XE}$ be a cq-state. We say that $X$ conditioned on $E$ has $s,\varepsilon$ quantum guessing pseudoentropy $H_{\guess}^{s,\varepsilon}(X|E)_{\rho}\ge k$ if for every circuit $\cC$ of size at most $s$ the probability of guessing $X$ correctly from $E$ is
    \begin{equation*}
        \prs{x\in X}{\cC(\rho^{x}_{E}) = x} \le 2^{-k} +\varepsilon\;.
    \end{equation*}
\end{definition}
Since the purified distance bounds the statistical distinguishability of states, we can see that
\begin{equation*}
    H_{\guess}^{s,\varepsilon}(X|E)_{\rho} \ge \hunps^{\varepsilon}(X|E)_{\rho} \;.
\end{equation*}
For $\varepsilon=0$, the definitions coincide for any $s\in\mathbb{N}$ and any cq-state
\begin{equation*}
    H_{\guess}^{s,0}(X|E)_{\rho} = \hunps^{0}(X|E)_{\rho} \;.
\end{equation*}
We do not know of a bound in the other direction between guessing pseudoentropy and other quantum computational entropies for any positive $\varepsilon$.

We write here a generalization of the classical definition of unpredictability~\cite{HLR07SepPseudoentropyfromCompressibility}, in quantum notations. We separate the computational parameter $s$ into two parameters to reflect its different roles in the definition.
\begin{definition}[Classical Conditional Unpredictability Entropy]\label{def:clasical_unpredictability_entropy}
    For any classical state $\rho_{XE}$, and $\varepsilon\ge 0$, $s_{\ind},s_{\guess}\in\bbN$.
    We say that
    \begin{equation*}
        \hunp^{\varepsilon,s_{\ind},s_{\guess}}(X|E)_{\rho}\ge k \;,
    \end{equation*}
    if there is $\tilde{\rho}_{XE}$ such that $\dist_{s_{\ind}}(\rho_{XE},\tilde{\rho}_{XE}) \le \varepsilon$, and for any guessing circuit $\cC$ of size $s_{\guess}$
    \begin{equation*}
        \pr{\cC(\tilde{\rho}_{E}^x) = x} \le 2^{-k} \;.
    \end{equation*}
\end{definition}

In the limit $(s_{\guess},s_{\ind}) \to \infty$, we would expect to recover the information-theoretic smooth min-entropy. It turns out not to be the case. The indistinguishably part $s_{\ind}$ converges to the trace distance, not to the purified distance like the information-theoretic smooth min-entropy.
Purified distance has a few properties that are very useful for properties we want conditional quantum entropy to have.
To recover some of these properties in a computational setting, we leave only $s_{\guess}$ as our computational assumption. We replace the $(s_{\ind},\varepsilon)$ computational indistinguishability with the $\varepsilon$ information-theoretic purified distance.
We will highlight an essential difference in the proof of~\cref{lem:chain-rule-quantum-unpredictability-entropy}.

\subsection{Additional Facts and Proofs}\label{apendix_sec_more_proofs}

Proof of~\cref{lem:data_processing_inequality_unpredictability_entropy}, recall:
\begin{lemma}[Data-Processing Inequality]
    Let $\rho_{XE}$ be a cq-state, $s\in\mathbb{N}$, $\varepsilon\ge0$, let $\Phi_{E\to E'}$ be a quantum channel that can be implemented using a circuit of size $t$,
    \begin{equation*}
        \hunps^{\varepsilon}(X|E')_{\Phi(\rho)} \ge \hunps[s+t]^{\varepsilon}(X|E)_{\rho}\;.
    \end{equation*}
\end{lemma}

\begin{proof}
    Let $\rho_{XE}$ be a cq-state and let $\Phi_{E\to E'}$ be a quantum channel that is implemented by a circuit of size $t$. Fix any guessing circuit $\cC'$ acting on $E'$ with size $s$. Since $\Phi_{E\to E'}$ is realized by a circuit of size $t$, we can compose the guessing circuit $\cC'$ with the circuit for $\Phi_{E\to E'}$ to obtain a guessing circuit $\cC$ on $E$ of size at most $s+t$.

    Let $\tilde{\rho}_{XE}$ be a cq-state such that
    \begin{equation*}
        \bP\left(\rho_{XE},\tilde{\rho}_{XE}\right) \le \varepsilon\;.
    \end{equation*}
    By the monotonicity of the purified distance under quantum channels, we have
    \begin{equation*}
        \bP\left((\1_X \otimes\Phi)(\rho_{XE}),(\1_X \otimes \Phi)(\tilde{\rho}_{XE})\right) \le \varepsilon\;.
    \end{equation*}
    Then, by~\cref{def:purified-unpredictability-entropy} if
    \begin{equation*}
        \hunps[s+t](X|E)_{\rho} \ge k \;,
    \end{equation*}
    then
    \begin{equation*}
        \pr{\cC\left(\tilde{\rho}_E^x\right)=x} \le 2^{-k}\;.
    \end{equation*}
    Since $\cC'$ is obtained from $\cC$, we deduce that any circuit of size $s$ acting on $E'$ satisfies
    \begin{equation*}
        \pr{\cC'\left(\tilde{\rho}_{E'}^x\right)=x} \le 2^{-k}\;.
    \end{equation*}
    By the definition of $\hunps^{\varepsilon}(X|E')_{\Phi(\rho)}$ this implies that
    \begin{equation*}
        \hunps(X|E')_{\Phi(\rho)} \ge k \;,
    \end{equation*}
    or equivalently,
    \begin{equation*}
        \hunps(X|E')_{\Phi(\rho)} \ge \hunps[s+t](X|E)_{\rho}\;.  \qedhere
    \end{equation*}
\end{proof}

Proof of~\cref{lem:inequality-of-operators-for-unp-leakage-chain-rule}, from~\cite{WTHR11Impossibilitygrowingquantumbitcommit}, recall:
\begin{lemma}
    For any state $\rho_{A}$ and any extension $\rho_{AB}$, we have:
    \begin{equation*}
        \rho_{AB} \le \dim(B)^{2} (\rho_{A} \otimes \omega_{B}) \;,
    \end{equation*}
    where $\omega_{B}$ is the maximally mixed state on $B$.
\end{lemma}

\begin{proof}
    This lemma is proven as part of Lemma 12 in~\cite{WTHR11Impossibilitygrowingquantumbitcommit}. We give a more detailed proof here for convenience.

    Starting with the case of a pure extension.
    Let $\proj{\psi}_{AB}$ be a pure state. We define
    $$\tau_{A} = \ptr{B}{\proj{\psi}_{AB}}\;,$$
    $$\Gamma_{AB} = (\tau^{-\frac{1}{2}}_{A} \otimes \1_{B}) \proj{\psi}_{AB} (\tau^{-\frac{1}{2}}_{A} \otimes \1_{B}) \;, $$
    where the inverse is on the support of $\tau_{A}$.
    By construction $\Gamma_{AB}$ is a rank $1$ matrix, so $\lambda_{\max}(\Gamma_{AB}) = \tr{\Gamma_{AB}}$.
    From the Schmidt decomposition we can write $\ket{\psi}_{AB} = \sum_{i} \sqrt{p_{i}} \ket{a_{i}} \otimes \ket{b_{i}}$ where $\set{\ket{a_{i}}},\set{\ket{b_{i}}}$ are orthonormal bases for $A$ and $B$ respectively. We also get that $\tau_{A} = \sum_{i} p_{i} \proj{a_{i}}$, and $\tau_{A}^{-\frac{1}{2}} = \sum_{i} p_{i}^{-\frac{1}{2}} \proj{a_{i}}$ where the inverse applies for all $p_{i} > 0$, and $0$ otherwise. We can see that:
    \begin{align*}
        \tr{\Gamma_{AB}}
         & = \tr{\left( \sum_{i} p_{i}^{-\frac{1}{2}} \proj{a_{i}} \otimes \1_{B} \right)
        \left(\sum_{i} \sqrt{p_{i}} \ket{a_{i}} \otimes \ket{b_{i}} \right) \right. \nonumber                                   \\
         & \qquad \left. \left(\sum_{i} \sqrt{p_{i}} \bra{a_{i}} \otimes \bra{b_{i}}\right)
        \left( \sum_{i} p_{i}^{-\frac{1}{2}} \proj{a_{i}} \otimes \1_{B} \right)}                                               \\
         & = \tr{\sum_{i}\delta_{p_{i}}\ket{a_{i}} \otimes \ket{b_{i}} \sum_{i} \delta_{p_{i}} \bra{a_{i}} \otimes \bra{b_{i}}} \\
         & = \sum_{i} \delta_{p_i} = \rank{\tau_{A}} \;,
    \end{align*}
    where $\delta_{p_{i}} = 1$ if $p_{i} \ne 0$ and $0$ if $p_{i}=0$.
    We also get from the Schmidt decomposition that $\rank{\tau_{A}} = \rank{\tau_{B}}$ and therefore $\rank{\tau_{A}} \le \min \set{|A|,|B|}$. Combining the inequalities we get:
    \begin{equation*}
        \lambda_{\max}(\Gamma_{AB}) = \tr{\Gamma_{AB}} = \rank{\tau_{A}} \le \min\set{|A|,|B|} \;.
    \end{equation*}
    And so $\Gamma_{AB} \le |B| \1_{AB}$.
    Applying $\tau_{A}^{\frac{1}{2}} \otimes \1_{B}$ to both sides we get: $\proj{\psi}_{AB} \le |B| (\tau_{A} \otimes \1_{B})$, or equivalently $\proj{\psi}_{AB} \le |B|^{2} (\tau_{A} \otimes \omega_{B})$ which concludes the proof for pure states.

    Since mixed states are convex combinations of pure states, we get the same inequality for mixed states by taking the same convex combination on both sides of the inequality. For any state $\rho_{AB}$ we get
    \begin{equation*}
        \rho_{AB} \le \dim(B)^{2} (\rho_{A} \otimes \omega_{B}) \;.
    \end{equation*}
    In~\cite{WTHR11Impossibilitygrowingquantumbitcommit} they note that this holds for any weighted sum of pure states with positive weights, and therefore holds for any positive operator, and not just quantum states.%
\end{proof}

Proof of~\cref{lemma:triangle-inequality-computational-td}, recall
\begin{lemma}[Triangle Inequality for Computational Distance]
    For any $s\in\bbN$ and states $\rho,\sigma,\tau$:
    \begin{equation*}
        \dist_{s}(\rho,\sigma) \le \dist_{s}(\rho,\tau) + \dist_{s}(\tau,\sigma) \;.
    \end{equation*}
\end{lemma}

\begin{proof}
    For any states $\rho,\sigma,\tau$ and let $\cC$ be a fixed distinguisher with circuit of size at most $s$, from the triangle inequality:
    \begin{align*}
        \abs{\pr{\cC(\rho)=1} - \pr{\cC(\sigma)=1}}  \le
         & \abs{\pr{\cC(\rho)=1} - \pr{\cC(\tau)=1}}         \\
         & + \abs{\pr{\cC(\tau)=1} - \pr{\cC(\sigma)=1}} \;.
    \end{align*}
    In particular, denote $\cC$ the circuit that saturates the definition of the computational trace distance for $\rho$ and $\sigma$:
    \begin{align*}
        \dist_{s}(\rho,\sigma) = \abs{\pr{\cC(\rho)=1} - \pr{\cC(\sigma)=1}} \le
         & \abs{\pr{\cC(\rho)=1} - \pr{\cC(\tau)=1}}         \\
         & + \abs{\pr{\cC(\tau)=1} - \pr{\cC(\sigma)=1}} \;.
    \end{align*}
    By definition, since the computational distance is defined by the maximum in the set of all distinguishers with a circuit of size at most $s$:
    \begin{align*}
        \abs{\pr{\cC(\rho)=1} - \pr{\cC(\tau)=1}}   & \le \dist_{s}(\rho,\tau) \;,   \\
        \abs{\pr{\cC(\tau)=1} - \pr{\cC(\sigma)=1}} & \le \dist_{s}(\tau,\sigma) \;.
    \end{align*}
    Therefore:
    \begin{equation*}
        \dist_{s}(\rho,\sigma) \le \dist_{s}(\rho,\tau) + \dist_{s}(\tau,\sigma) \;. \qedhere
    \end{equation*}
\end{proof}

\bibliographystyle{IEEEtran}
\bibliography{referencefile}

\end{document}